\documentclass[11pt]{article}%{ctexart}

\usepackage{graphicx}
\usepackage[margin=0.5in]{geometry}
\usepackage{bm}
\usepackage{amsthm}
\usepackage{amscd}
\usepackage{amsbsy}
\usepackage{appendix}
\usepackage{amsmath}
\usepackage{amsfonts}
\usepackage{amssymb}
\usepackage{calligra}
\usepackage{xcolor}
\usepackage{float}
\usepackage[T1]{fontenc}
\usepackage{multirow}
\usepackage{mathrsfs}
\usepackage{mathtools}
\usepackage{epstopdf}
\usepackage{psfrag}
\usepackage{subfigure}
\usepackage{booktabs}
\usepackage{listings}
\usepackage{longtable}
\usepackage{latexsym}
\usepackage{arydshln}
\usepackage{enumerate}
\usepackage{epsfig}
\usepackage{cite}
\usepackage{verbatim}
\usepackage{overpic}
\usepackage{abstract}
\usepackage{extarrows}
\usepackage[pagewise]{lineno}%\linenumbers

\usepackage[noend]{algpseudocode}
\usepackage{algorithmicx,algorithm}

\usepackage[figuresleft]{rotating}

\allowdisplaybreaks[4]

\date{}
\title{Using Symbolic Computation to Analyze Zero-Hopf Bifurcations of Polynomial Differential Systems}
\author{Bo Huang\\
	\it\footnotesize LMIB -- School of Mathematical Sciences,
Beihang University, Beijing 100191, China \\
	\it\footnotesize bohuang0407@buaa.edu.cn}

\newtheorem {theorem*}{Theorem}
\newtheorem{theorem} {Theorem}

\newtheorem{lemma}{Lemma}

\newtheorem{remark}{Remark}

\newtheorem{open problem} {Open problem}

\numberwithin{equation}{section}

\usepackage[colorlinks,
            CJKbookmarks=true,
            linkcolor=blue,
            anchorcolor=red,
             urlcolor=blue,
            citecolor=blue
            ]{hyperref}

\begin{document}
\maketitle
\noindent {\bf Abstract.}  This paper is devoted to the study of infinitesimal limit cycles that can bifurcate from zero-Hopf equilibria of differential systems based on the averaging method. We develop an efficient symbolic program using Maple for computing the averaged functions of any order for continuous differential systems in arbitrary dimension. The program allows us to systematically analyze zero-Hopf bifurcations of polynomial differential systems using symbolic computation methods. We show that for the first-order averaging, $\ell\in\{0,1,\ldots,2^{n-3}\}$ limit cycles can bifurcate from the zero-Hopf equilibrium for the general class of perturbed differential systems and up to the second-order averaging, the maximum number of limit cycles can be determined by computing the mixed volume of a polynomial system obtained from the averaged functions. A number of examples are presented to demonstrate the effectiveness of the proposed algorithmic approach.

\smallskip

\noindent {\bf Math Subject Classification (2020).} 34C07; 37G15; 68W30.

\smallskip


\section{Introduction}
%Bounding the number of limit cycles for systems of polynomial differential equations is a long standing problem in the field of dynamical systems.

Many real-world phenomena are modeled using autonomous systems of parametric differential equations. In this paper, we deal with polynomial differential systems in $\mathbb{R}^n$ of the form
\begin{equation}\label{eq1-0}
\begin{split}
\dot{\boldsymbol{x}}=\boldsymbol{f}(\boldsymbol{x},\boldsymbol{\mu}),\quad \boldsymbol{f}:\mathbb{R}^n\times\mathbb{R}^p\rightarrow\mathbb{R}^n,
\end{split}
\end{equation}
where $\boldsymbol{x}=(x_1,\ldots,x_n)$ are variables, $\boldsymbol{\mu}=(\mu_1,\ldots,\mu_p)$ are real parameters, and $\boldsymbol{f}=(f_1,\ldots,f_n)$ are polynomials in $\mathbb{R}[\boldsymbol{x}]$.

Bifurcation analysis for differential systems of the form \eqref{eq1-0} is a central problem that has been extensively explored in the theory of dynamical systems. The analysis of bifurcations of dynamical systems usually involves heavy computations with large symbolic expressions which can neither be performed numerically using limited precision, nor be done effectively by hand for the complexity in question. In the past few decades, symbolic methods have been explored extensively in terms of the qualitative analysis of dynamical systems. To a large extent, the tedious deduction problems in qualitative analysis have been resolved and many encouraging results have been obtained regarding the stability analysis of dynamical systems \cite{HLS97,WDMXBC2005,HTX15}, the determination of
center conditions \cite{dw91,vd09,acj17}, the estimation of the number of limit cycles \cite{HGR10,CCMYZ13,SH2017}, the investigation of invariant algebraic curves \cite{ykm1993,bhlr2015,adrt2016}, etc. More recent progress on symbolic computation methods for the qualitative theory of differential equations can be found in \cite{HNW2022}.

In this paper, we study the zero-Hopf bifurcations of polynomial differential systems of the form \eqref{eq1-0}. More precisely, \textit{we would like to compute a partition of the parametric space of $\boldsymbol{\mu}$ such that, inside every open cell of the partition, the system can have the maximum number of limit cycles bifurcating from a zero-Hopf equilibrium point}. The main techniques are based on the averaging method and some algebraic methods with exact symbolic computation. The zero-Hopf bifurcation in three-dimensional case has been extensively studied in the literature (see \cite{GH93,YK04,jlcv2011,EJL17} and references therein); and it has been shown that some complicated invariant sets of the unfolding could be bifurcated from the isolated zero-Hopf equilibrium under some conditions. Hence, in some cases zero-Hopf bifurcation could imply a local birth of ``chaos'' (see, e.g., \cite{SJM84,BMS06}). The ideas of using the averaging method to study zero-Hopf bifurcations of nonlinear differential systems are already presented in several papers (see \cite{JLAM2016,BLM2016,EJL17,JA20,JLXZ09,dpxz2009,lzhang2011,BLV18,jlmrc2018,jlzy2021} and references therein). Here we summarize the used techniques and introduce an algorithmic approach for systematically analyzing zero-Hopf bifurcations by using symbolic computation. It should be noted that center manifold theory and normal form theory are also powerful tools for the analysis of zero-Hopf bifurcations of nonlinear differential systems (see, e.g., \cite{GH93,YK04,HY12}). Algorithms have been developed for the computation of center manifolds and normal forms \cite{BiYu99,TiYu14}, but they give no qualitative information about the bifurcated limit cycles. In contrast, using averaged functions, one can determine the shape of the bifurcated limit cycles up to any order in the parameter $\varepsilon$ (see, e.g., \cite{JLXZ09,LMB09}). On the other hand, it has been shown in \cite{BPY20} that the averaging method may be unable to detect possible limit cycles bifurcating from a zero-Hopf equilibrium, while the normal form theory can be used to overcome the difficulty.

Recall that a limit cycle of differential system \eqref{eq1-0} is an isolated periodic orbit of the system. The method of averaging is one of the best analytical methods to study limit cycles of differential systems in the presence of a small parameter $\varepsilon$ (see \cite{LMS15,LNT14} and references therein). Roughly speaking, in averaging theory one replaces a vector field by its average (over time or an angular variable) with the goal of obtaining asymptotic approximations to the original system that will be capable of guaranteeing the existence of limit cycles. The averaging method provides a straightforward calculation approach for determining the number of limit cycles of the regarded particular class of differential systems. However, in practice, the evaluation of the averaged functions is a computational problem that requires powerful computerized resources. Moreover, the computational complexity grows very fast when the averaging order, degree, or dimension of the required systems increases. The aim of this paper is twofold: first, we develop an efficient symbolic program using Maple for computing the averaged functions of any order for continuous differential systems in arbitrary dimension. Second, we use symbolic computation methods (combined with the program) to provide a systematical and algorithmic approach for analyzing zero-Hopf bifurcations of polynomial differential systems.

%In our recent work \cite{hy19}, we provide an algorithmic approach to small limit cycles of planar differential systems by the averaging method, and give an upper bound for the number of zeros of the average functions for general perturbed differential systems (\cite{hy19}, Thm. 3.1). The second goal of this paper is to extend our algorithmic approach to the analysis of zero-Hopf bifurcations in higher dimensional polynomial differential systems.

%The method has a long history that started with the classical works of Lagrange and Laplace, who provided an intuitive justification of the method. The first formalization of this theory was done in 1928 by Fatou. Important practical and theoretical contributions to the averaging method were made in the 1930's by Bogoliubov-Krylov, and in 1945 by Bogoliubov. The ideas of averaging method have extended in several directions for finite and infinite dimensional differentiable systems.

%In general, to obtain analytically periodic solutions of a differential system is a very difficult problem, many times a problem impossible to solve. As we shall see when we can apply the averaging method in the study of zero-Hopf bifurcations, this difficult problem is reduced to finding the number of common zeros of a polynomial system.

This paper is organized as follows. In Section \ref{sect2}, we first introduce the concept of zero-Hopf bifurcation of differential systems and then present the basic theory of the averaging method for studying limit cycles. Section \ref{main-1} contains the main results, such as Lemma \ref{lem-main-1}, which gives the parametric formula of the standard form of averaging associated to system \eqref{equ-2}; and Theorem \ref{semi-averaging}, which can be used to derive sufficient conditions for a given parametric differential system of the form \eqref{equ-2} to have a prescribed number of limit cycles. In Section \ref{sect3}, we present the outline of the Maple program for computing the averaged functions \eqref{equ2-3} and propose an algorithmic approach for automatically deriving Theorem \ref{semi-averaging}. The effectiveness of our computational approach is demonstrated in Section \ref{sect4} by using results obtained for a famous jerk differential system, a class of generalized Lorenz systems, and a 4D hyperchaotic differential system. The paper is concluded with a few remarks.

The proofs of Lemma \ref{lem-main-1}, Theorem \ref{main-th-1}, and Theorem \ref{main-2} are given in appendices. The Maple program developed in this paper can be found at{\small{
{\textcolor{blue}{\underline{\url{https://github.com/Bo-Math/zero-Hopf}}}}}}.

\section{Zero-Hopf Bifurcation and the Theory of Averaging}\label{sect2}
This section is devoted to the concept of zero-Hopf bifurcation of differential systems and the basic theory of the averaging method for studying limit cycles. We assume that a polynomial differential system of the form \eqref{eq1-0} has a singularity at the origin. Let $J$ be the Jacobian matrix of the associated linearized system at the origin. When $J$ has a pair of purely imaginary eigenvalues (e.g., $\pm bi$) and the other eigenvalues are nonzero, we call the origin a Hopf equilibrium point. If some of the eigenvalues of $J$ other than $\pm bi$ are zero, then the origin is called a zero-Hopf equilibrium point. Limit cycles may bifurcate from Hopf equilibria in nonlinear differential systems of the form \eqref{eq1-0} as the values of the parameters in the coefficients vary. A generic Hopf bifurcation is a bifurcation of a limit cycle from a Hopf equilibrium and a zero-Hopf bifurcation is a bifurcation of a limit cycle from a zero-Hopf equilibrium. Here we are interested in zero-Hopf bifurcations where all the eigenvalues of $J$ different from $\pm bi$ are zero; such kind of zero-Hopf bifurcation is called \textit{complete} zero-Hopf bifurcation.

The investigations in this paper are focused on \textit{complete} zero-Hopf bifurcations for differential systems of the form \eqref{eq1-0}. In this case, system \eqref{eq1-0} of degree at most $N$ can be written in the form
\begin{equation}\label{equ-1}
\begin{split}
\dot{x}_1&=f_1=-bx_2+\sum_{m\geq2}^{N}\sum_{i_1+\cdots+i_n=m}a_{i_1,\ldots,i_n}x_1^{i_1}\cdots x_n^{i_n},\\
\dot{x}_2&=f_2=bx_1+\sum_{m\geq2}^{N}\sum_{i_1+\cdots+i_n=m}b_{i_1,\ldots,i_n}x_1^{i_1}\cdots x_n^{i_n},\\
\dot{x}_s&=f_s=\sum_{m\geq2}^{N}\sum_{i_1+\cdots+i_n=m}c_{i_1,\ldots,i_n,s}x_1^{i_1}\cdots x_n^{i_n},
\end{split}
\end{equation}
where $s=3,\ldots,n$, $b\neq0$, $a_{i_1,i_2,\ldots,i_n}$, $b_{i_1,i_2,\ldots,i_n}$ and $c_{i_1,i_2,\ldots,i_n,s}$ are real parameters.

Our objective is to determine how many limit cycles can bifurcate from the origin, as a zero-Hopf equilibrium of system \eqref{equ-1}, when the system is perturbed inside the class of differential systems of the same form. We shall use the method of averaging \cite{LNT14}, up to an arbitrary order in $\varepsilon$ for studying limit cycles of differential systems. To do this, one usually considers the following perturbations of system \eqref{equ-1}
\begin{equation}\label{equ-2}
\begin{split}
\dot{x}_1&=f_1+p_1(x_1,\ldots,x_n,\varepsilon),\\
\dot{x}_2&=f_2+p_2(x_1,\ldots,x_n,\varepsilon),\\
\dot{x}_s&=f_s+p_s(x_1,\ldots,x_n,\varepsilon),\quad s=3,\ldots,n,
\end{split}
\end{equation}
where
\begin{equation}
\begin{split}
p_1&=\sum_{j=1}^k\varepsilon^j\sum_{m^*\geq1}^{N}\sum_{i_1+\cdots+i_n=m^*}a_{j,i_1,\ldots,i_n}x_1^{i_1}\cdots x_n^{i_n},\\
p_2&=\sum_{j=1}^k\varepsilon^j\sum_{m^*\geq1}^{N}\sum_{i_1+\cdots+i_n=m^*}b_{j,i_1,\ldots,i_n}x_1^{i_1}\cdots x_n^{i_n},\\
p_s&=\sum_{j=1}^k\varepsilon^j\sum_{m^*\geq1}^{N}\sum_{i_1+\cdots+i_n=m^*}c_{j,i_1,\ldots,i_n,s}x_1^{i_1}\cdots x_n^{i_n},\nonumber
\end{split}
\end{equation}
the constants $a_{j,i_1,\ldots,i_n}$, $b_{j,i_1,\ldots,i_n}$ and $c_{j,i_1,\ldots,i_n}$ are real, and $\varepsilon$ is a small parameter.

Since the eigenvalues of the linearization of system \eqref{equ-2} at the origin when $\varepsilon=0$ are $\pm bi\neq0$ and 0 with multiplicity $n-2$, if one or several infinitesimal periodic orbits of system \eqref{equ-2} bifurcate from the origin at $\varepsilon=0$, we see that such kind of bifurcation is then a \textit{complete} zero-Hopf bifurcation. In this paper, by using the averaging method, we are interested in the maximum number of limit cycles of system \eqref{equ-2} for $|\varepsilon|$ sufficiently small, which bifurcate from the origin in such a zero-Hopf bifurcation.

The averaging method deals with differential systems in the following standard form
\begin{equation}\label{equ2-1}
\begin{split}
\frac{d\boldsymbol{x}}{dt}=\sum_{i=0}^k\varepsilon^i\boldsymbol{F}_i(t,\boldsymbol{x})+\varepsilon^{k+1}\boldsymbol{R}(t,\boldsymbol{x},\varepsilon),
\end{split}
\end{equation}
where $\boldsymbol{F}_i:\mathbb{R}\times D\rightarrow\mathbb{R}^n$ for $i=0,1,\ldots,k$, and $\boldsymbol{R}:\mathbb{R}\times D\times(-\varepsilon_0,\varepsilon_0)\rightarrow\mathbb{R}^n$ are continuous functions, and $T$-periodic in the variable $t$, $D$ being an open subset of $\mathbb{R}^n$.

We first introduce some notations before presenting the main results about this method. Let $L$ be a positive integer, $\boldsymbol{x}=(x_1,\ldots,x_n)\in D$, $t\in\mathbb{R}$ and $\boldsymbol{y}_j=(y_{j1},\ldots,y_{jn})\in\mathbb{R}^n$ for $j=1,\ldots,L$. Given $\boldsymbol{F}:\mathbb{R}\times D\rightarrow\mathbb{R}^n$ a sufficiently smooth function, for each $(t,\boldsymbol{x})\in\mathbb{R}\times D$ we denote by $\partial^L\boldsymbol{F}(t,\boldsymbol{x})$ a symmetric $L$-multilinear map which is applied to a `product' of $L$ vectors of $\mathbb{R}^n$, which we denote as $\bigodot_{j=1}^L\boldsymbol{y}_j\in\mathbb{R}^{nL}$. The definition of this $L$-multilinear map is
\begin{equation}\label{equ2-2}
\begin{split}
\partial^L\boldsymbol{F}(t,\boldsymbol{x})\bigodot_{j=1}^L\boldsymbol{y}_j=\sum_{i_1,\ldots,i_L=1}^n
\frac{\partial^L\boldsymbol{F}(t,\boldsymbol{x})}{\partial x_{i_1}\cdots\partial x_{i_L}}y_{1i_1}\cdots y_{Li_L}.
\end{split}
\end{equation}

\begin{remark}\label{remk2-1}
The $L$-multilinear map defined in \eqref{equ2-2} is the $L$th Fr\'echet derivative of the function $\boldsymbol{F}(t,\boldsymbol{x})$ with respect to the variable $\boldsymbol{x}$. Given a positive integer $b$ and a vector $\boldsymbol{y}\in\mathbb{R}^n$, we also denote $\boldsymbol{y}^b=\bigodot_{i=1}^b\boldsymbol{y}\in\mathbb{R}^{nb}$.
\end{remark}

The averaging method consists in defining a collection of functions $\boldsymbol{f}_i:D\rightarrow\mathbb{R}^n$, called the $i$th-order averaged function, for $i=1,2,\ldots,k$, whose isolated zeros control, for $\varepsilon$ sufficiently small, the limit cycles of system \eqref{equ2-1}. In Llibre-Novaes-Teixeira \cite{LNT14} it has been established that
\begin{equation}\label{equ2-3}
\begin{split}
\boldsymbol{f}_i(\boldsymbol{z})=\frac{\boldsymbol{y}_i(T,\boldsymbol{z})}{i!},
\end{split}
\end{equation}
where $\boldsymbol{y}_i:\mathbb{R}\times D\rightarrow\mathbb{R}^n$, for $i=1,2,\ldots,k$, are defined recurrently by the following integral equations
\begin{equation}\label{equ2-4}
\begin{split}
\boldsymbol{y}_1(t,\boldsymbol{z})&=\int_0^t\boldsymbol{F}_1(\theta,\boldsymbol{z})d\theta,\\
\boldsymbol{y}_i(t,\boldsymbol{z})&=i!\int_0^t\Big(\boldsymbol{F}_i(\theta,\boldsymbol{z})
+\sum_{\ell=1}^{i-1}\sum_{S_{\ell}}\frac{1}{b_1!b_2!2!^{b_2}\cdots b_{\ell}!\ell!^{b_{\ell}}}\times\\
&\partial^L\boldsymbol{F}_{i-\ell}(\theta,\boldsymbol{z})
\bigodot_{j=1}^{\ell}\boldsymbol{y}_j(\theta,\boldsymbol{z})^{b_j}\Big)d\theta.
\end{split}
\end{equation}
Here $S_{\ell}$ is the set of all $\ell$-tuples of nonnegative integers $[b_1,b_2,\ldots,b_{\ell}]$ satisfying $b_1+2b_2+\cdots+\ell b_{\ell}=\ell$ and $L=b_1+b_2+\cdots+b_{\ell}$. For sake of simplicity, in \eqref{equ2-4}, we are assuming that $\boldsymbol{F}_0=0$. Remark that, related to the averaged functions \eqref{equ2-3} there exist two fundamentally different cases of system \eqref{equ2-1}, namely, when $\boldsymbol{F}_0=0$ and when $\boldsymbol{F}_0\neq0$. We see that when $\boldsymbol{F}_0\neq0$, the formula for $\boldsymbol{y}_i(t,\boldsymbol{z})$ requires the solution of a Cauchy problem for $i=1,2,\ldots,n$ (see \cite[Remark 3]{LNT14}). The investigation in this paper is restricted to the case where $\boldsymbol{F}_0=0$.

The following averaging theorem provides a criterion for the existence of limit cycles. Its proof can be found in \cite{LNT14}.
\begin{theorem}\label{averaging-thm}
Assuming the following conditions:
\begin{enumerate}
  \item for each $i=1,2,\ldots,k$ and $t\in\mathbb{R}$, the function $\boldsymbol{F}_i(t,\boldsymbol{x})$ is of class $\mathcal{C}^{k-i}$, $\partial^{k-i}\boldsymbol{F}_i$ is locally Lipschitz in $\boldsymbol{x}$, and $\boldsymbol{R}$ is a continuous function locally Lipschitz in $\boldsymbol{x}$;
  \item for some $j\in\{1,2,\ldots,k\}$, $\boldsymbol{f}_i=0$ for $i=1,2,\ldots,j-1$ and $\boldsymbol{f}_j\neq0$;
  \item for some $\boldsymbol{z}^*\in D$ with $\boldsymbol{f}_j(\boldsymbol{z}^*)=0$ we have ${\rm{det}}(J_{\boldsymbol{f}_j}(\boldsymbol{z}^*))\neq0$.
      %that $d_B(\boldsymbol{f}_j(\boldsymbol{z}),V,0)\neq0$.
\end{enumerate}
Then, for $|\varepsilon|>0$ sufficiently small, there exists a $T$-periodic solution $\boldsymbol{x}(t,\varepsilon)$ of \eqref{equ2-1} such that $\boldsymbol{x}(0,\varepsilon)\rightarrow\boldsymbol{z}^*$ when $\varepsilon\rightarrow0$.
\end{theorem}

\begin{remark}\label{remk2-2}
The notation $\mbox{det}(J_{\boldsymbol{f}_j}(\boldsymbol{z}))\neq0$ means that the Jacobian determinant of $\boldsymbol{f}_j$ at $\boldsymbol{z}\in V$ is nonzero. In Theorem \ref{averaging-thm}, the function $\boldsymbol{f}_j$ for $j\in\{1,2,\ldots,k\}$ is assumed to be a $\mathcal{C}^1$ function. In this case, instead of Brouwer degree theory, the implicit function theorem could be used to
prove Theorem \ref{averaging-thm}, see \cite[Remark 4]{LNT14}.

%$\mbox{det}(J_{\boldsymbol{f}_j}(\boldsymbol{z}^*))\neq0$ implies $d_B(\boldsymbol{f}_j(\boldsymbol{z}),V,0)\neq0$. The symbol $d_B$ is the Brouwer degree; see \cite{feb83} for more details.
\end{remark}

In practical terms, the evaluation of the recurrence \eqref{equ2-4} is a computational problem. Recently, Novaes \cite{dn17} used the partial Bell polynomials to provide an alternative formula for the recurrence \eqref{equ2-4}. This formula can make the computational implementation of the averaged functions \eqref{equ2-3} easier. A partial Bell polynomial is expressed by
\begin{equation}\label{equ2-5}
\begin{split}
B_{\ell,m}(x_1,\ldots,x_{\ell-m+1})=\sum_{\widetilde{S}_{\ell,m}}\frac{\ell!}{b_1!b_2!\cdots b_{\ell-m+1}!}\prod_{j=1}^{\ell-m+1}\left(\frac{x_j}{j!}\right)^{b_j},
\end{split}
\end{equation}
where $\ell$ and $m$ are positive integers, $\widetilde{S}_{\ell,m}$ is the set of all $(\ell-m+1)$-tuples of nonnegative integers $[b_1,b_2,\ldots,b_{\ell-m+1}]$ satisfying $b_1+2b_2+\cdots+(\ell-m+1)b_{\ell-m+1}=\ell$, and $b_1+b_2+\cdots+b_{\ell-m+1}=m$.

The following result provides an alternative formula for the averaged functions.

\begin{theorem}\label{thm2-a}
For $i=1,2,\ldots,k$ the recurrence \eqref{equ2-4} reads
\begin{equation}\label{equ2-6}
\begin{split}
\boldsymbol{y}_1(t,\boldsymbol{z})&=\int_0^t\boldsymbol{F}_1(\theta,\boldsymbol{z})d\theta,\\
\boldsymbol{y}_i(t,\boldsymbol{z})&=i!\int_0^t\Big(\boldsymbol{F}_i(\theta,\boldsymbol{z})+\sum_{\ell=1}^{i-1}\sum_{m=1}^{\ell}\frac{1}{\ell!}\times\\
&\partial^m\boldsymbol{F}_{i-\ell}(\theta,\boldsymbol{z})
B_{\ell,m}\big(\boldsymbol{y}_1(\theta,\boldsymbol{z}),\ldots,\boldsymbol{y}_{\ell-m+1}(\theta,\boldsymbol{z})\big)\Big)d\theta.
\end{split}
\end{equation}

\end{theorem}

\begin{remark}\label{remk2-3}
Theorem \ref{thm2-a} gives an equivalent formulation of the averaged functions via Bell polynomial. But, this equivalent formulation is just focused on the form of the expressions, can not easily be applied to the computation of averaged functions for high-dimensional differential systems due to the operation of the $L$-multilinear map (see \eqref{equ2-2}) is missed. In other words, the $m$-multilinear map in \eqref{equ2-6} should be applied to each term of the Bell polynomial $B_{\ell,m}\big(\boldsymbol{y}_1(\theta,\boldsymbol{z}),\ldots,\boldsymbol{y}_{\ell-m+1}(\theta,\boldsymbol{z})\big)$. For this reason, the symbolic program we develop in this paper is based on both the $L$-multilinear map and Bell polynomials, which can be used efficiently for generating the integral equations \eqref{equ2-6}.
\end{remark}

Note that, under the assumptions of Theorem \ref{averaging-thm}, one can obtain the following expression of the limit cycle associated to the zero $\boldsymbol{z}^*$ of $\boldsymbol{f}_j(\boldsymbol{z})$:
\begin{equation}\label{equ2-9}
\begin{split}
\boldsymbol{x}(t,\varepsilon)&=\boldsymbol{z}^*+\varepsilon \boldsymbol{y}_1(t,\boldsymbol{z}^*)+\varepsilon^2\frac{\boldsymbol{y}_2(t,\boldsymbol{z}^*)}{2!}+\cdots\\
&\quad+\varepsilon^j \frac{\boldsymbol{y}_j(t,\boldsymbol{z}^*)}{j!}+\mathcal{O}(\varepsilon^{j+1}).
\end{split}
\end{equation}
See \cite{LNT14} for a proof of this result. More detailed discussions
of the averaging method, including applications, can be found in \cite{SVM2007,LMS15}.

The averaging method allows one to find periodic solutions of periodic nonautonomous differential systems (see \eqref{equ2-1}). However we are interested in using it for analyzing limit cycles bifurcating from a zero-Hopf equilibrium of the autonomous differential system \eqref{equ-2}. The general study of the exact number of isolated zeros of the averaged functions \eqref{equ2-3} up to every order is usually very difficult to be done, since the averaged functions may be too complicated, such as including square root functions, logarithmic functions, and elliptic integrals. But for zero-Hopf bifurcation analysis, the main work here is to study the maximum number of real solutions of a resulting polynomial system obtained from the averaged functions. Some advanced techniques from symbolic computation, such as Gr\"obner basis \cite{BB85}, triangular decomposition \cite{WTW00,DW01}, quantifier elimination \cite{GEC75,CHH91}, and real solution classification \cite{YHX01,LYBX05} may be used to perform the task.

\section{Theoretical Results on the Number of Limit Cycles}\label{main-1}

In this section, we present the main results on the number of limit cycles of system \eqref{equ-2} obtained by using symbolic computation methods. To study the zero-Hopf bifurcation for perturbed system \eqref{equ-2}, we should perform several changes of coordinates to transform the perturbed system into the standard form of averaging. The main results in this paper are based on the following Lemma \ref{lem-main-1}. Its proof can be found in Appendix \ref{sec-A}.

\begin{lemma}\label{lem-main-1}
The parametric formula of the standard form of averaging associated to system \eqref{equ-2} is as follows:
\begin{equation}\label{equ3-0-1}
\begin{split}
\frac{dR}{d\theta}&=\frac{R\big[\cos\theta\cdot S_{1}(\theta,R,X_3,\ldots,X_n,\varepsilon)+\sin\theta\cdot S_{2}(\theta,R,X_3,\ldots,X_n,\varepsilon)\big]}{bR+\cos\theta\cdot S_{2}(\theta,R,X_3,\ldots,X_n,\varepsilon)-\sin\theta\cdot S_{1}(\theta,R,X_3,\ldots,X_n,\varepsilon)},\\
&=\varepsilon F_{1,1}(\theta,R,X_3,\ldots,X_n)+\varepsilon^2F_{2,1}(\theta,R,X_3,\ldots,X_n)\\
&\quad+\cdots+\varepsilon^kF_{k,1}(\theta,R,X_3,\ldots,X_n)+\mathcal{O}(\varepsilon^{k+1}),\\
\frac{dX_s}{d\theta}&=\frac{R\cdot S_{s}(\theta,R,X_3,\ldots,X_n,\varepsilon)}{bR+\cos\theta\cdot S_{2}(\theta,R,X_3,\ldots,X_n,\varepsilon)-\sin\theta\cdot S_{1}(\theta,R,X_3,\ldots,X_n,\varepsilon)},\\
&=\varepsilon F_{1,s}(\theta,R,X_3,\ldots,X_n)+\varepsilon^2F_{2,s}(\theta,R,X_3,\ldots,X_n)\\
&\quad+\cdots+\varepsilon^kF_{k,s}(\theta,R,X_3,\ldots,X_n)+\mathcal{O}(\varepsilon^{k+1}),
\quad s=3,\ldots,n,
\end{split}
\end{equation}
where
%\begin{small}
\begin{equation}
\begin{split}
&S_{1}=\sum_{m\geq2}^N\varepsilon^{m-1}\sum_{i_1+\cdots+i_n=m}a_{i_1,\ldots,i_n}\Psi(\theta,R,X_3,\ldots,X_n)\\
&+\sum_{j=1}^k\varepsilon^j\sum_{m^*\geq1}^N\varepsilon^{m^*-1}\sum_{i_1+\cdots+i_n=m^*}a_{j,i_1,\ldots,i_n}\Psi(\theta,R,X_3,\ldots,X_n),\\
&S_{2}=\sum_{m\geq2}^N\varepsilon^{m-1}\sum_{i_1+\cdots+i_n=m}b_{i_1,\ldots,i_n}\Psi(\theta,R,X_3,\ldots,X_n)\\
&+\sum_{j=1}^k\varepsilon^j\sum_{m^*\geq1}^N\varepsilon^{m^*-1}\sum_{i_1+\cdots+i_n=m^*}b_{j,i_1,\ldots,i_n}\Psi(\theta,R,X_3,\ldots,X_n),\\
&S_{s}=\sum_{m\geq2}^N\varepsilon^{m-1}\sum_{i_1+\cdots+i_n=m}c_{i_1,\ldots,i_n,s}\Psi(\theta,R,X_3,\ldots,X_n)\\
&+\sum_{j=1}^k\varepsilon^j\sum_{m^*\geq1}^N\varepsilon^{m^*-1}
\sum_{i_1+\cdots+i_n=m^*}c_{j,i_1,\ldots,i_n,s}\Psi(\theta,R,X_3,\ldots,X_n)\nonumber
\end{split}
\end{equation}
%\end{small}
with $\Psi(\theta,R,X_3,\ldots,X_n)=(R\cos\theta)^{i_1}(R\cos\theta)^{i_2}X_3^{i_3}\cdots X_n^{i_n}$. The functions $F_{j_1,j_2}(\theta,R,X_3,\ldots,X_n)$, $j_1=1,2,\ldots,k$, $j_2=1,3,\ldots,n$ are obtained by carrying Taylor expansion of expressions in \eqref{equ3-0-1} with respective to the variable $\varepsilon$ around $\varepsilon=0$.
\end{lemma}

\begin{remark}
Note that the expression $S_j$ for $j=1,\ldots,n$ is a polynomial in $\varepsilon$ with degree at most $k+N-1$, and without constant term. So, after carrying Taylor expansion of expressions in \eqref{equ3-0-1} around $\varepsilon=0$, all the functions $F_{j_1,j_2}(\theta,R,X_3,\ldots,X_n)$, $j_1=1,2,\ldots,k$, $j_2=1,3,\ldots,n$ are rational functions in the variables $\theta,R,X_3,\ldots,X_n$. Moreover, the denominator of $F_{j_1,j_2}(\theta,R,X_3,\ldots,X_n)$ is a monomial with only the variable $R$. As a result, each of the integrand equation in \eqref{equ2-4} (or \eqref{equ2-6}) is a polynomial in $\sin\theta$ and $\cos\theta$. Hence, the resulting averaged function $\boldsymbol{f}_i=(f_{i,1},f_{i,3},\ldots,f_{i,n})$, for $i=1,\ldots,k$, is a collection of rational functions in the variables $R,X_3,\ldots,X_n$. One can check that, for each of the rational function $f_{i,j_2}$, $j_2=1,3,\ldots,n$, its denominator is a monomial with only the variable $R$.
\end{remark}

We denote by $H_k(n,N)$ the exact maximum number of limit cycles of system \eqref{equ-2}, which can bifurcate from the origin up to the $k$th order averaging. The following result shows that $H_1(n,N)=2^{n-3}$.

\begin{theorem}\label{main-th-1}
For $k=1$ and $|\varepsilon|>0$ sufficiently small, there are systems of the form \eqref{equ-2} having exactly $\ell\in\{0,1,\ldots,2^{n-3}\}$ limit cycles bifurcating from the origin at $\varepsilon=0$.
\end{theorem}

\begin{remark}\label{remk-th-1}
Theorem \ref{main-th-1} is a generalized version of Corollary 2 in \cite{JLXZ09}, its proof can be done by using similar techniques as in \cite{JLXZ09}. We put the detailed proof in Appendix \ref{A-B} for the completeness. Up to now and as far as we know, Theorem \ref{main-th-1} is the first result about the limit cycles on the general class perturbations of system \eqref{equ-2}. The number of limit cycles of system \eqref{equ-2} for the first-order averaging does not depend on $N$. For the study of special forms of system \eqref{equ-2}, we refer the reader to \cite{JLAM2016,BLM2016,EJL17,JA20}. Several studies on homogeneous perturbations can be found in \cite{JLXZ09,BLV18}.
\end{remark}

Let $\boldsymbol{\eta}=(R,X_3,\ldots,X_n)$, and $\boldsymbol{f}_k(\boldsymbol{\eta})=(f_{k,1}(\boldsymbol{\eta}),f_{k,3}(\boldsymbol{\eta}),\ldots,$
$f_{k,n}(\boldsymbol{\eta}))$ be the $k$th-order averaged functions associated to system \eqref{equ3-0-1}. In the following we provide an algorithmic approach for determining the number $H_k(n,N)$ by using methods of symbolic computation. Denote the Jacobian of the function $\boldsymbol{f}_k(\boldsymbol{\eta})$ by $J_{\boldsymbol{f}_k}(\boldsymbol{\eta})$. That is,
\begin{equation}\label{equ3-0-3}
\begin{split}
J_{\boldsymbol{f}_k}(\boldsymbol{\eta})=
\left[
 \begin{matrix}
   \frac{\partial f_{k,1}}{\partial R} & \frac{\partial f_{k,1}}{\partial X_3} & \cdots& \frac{\partial f_{k,1}}{\partial X_n} \\
   \frac{\partial f_{k,3}}{\partial R} & \frac{\partial f_{k,3}}{\partial X_3} & \cdots& \frac{\partial f_{k,3}}{\partial X_n} \\
   \vdots& \vdots & \ddots & \vdots \\
   \frac{\partial f_{k,n}}{\partial R} & \frac{\partial f_{k,n}}{\partial X_3} & \cdots& \frac{\partial f_{k,n}}{\partial X_n} \\
  \end{matrix}
  \right].\nonumber
\end{split}
\end{equation}

Let $D_k(\boldsymbol{\eta})$ be the determinate of the Jacobian $J_{\boldsymbol{f}_k}(\boldsymbol{\eta})$. The following result provides sufficient conditions for system \eqref{equ-2} to have exactly $\ell$ limit cycles bifurcating from the origin.

\begin{theorem}\label{semi-averaging}
For $|\varepsilon|>0$ sufficiently small, system \eqref{equ-2} up to the $k$th-order averaging has exactly $\ell$ limit cycles bifurcating from the origin if the following semi-algebraic system
\begin{equation}\label{equ3-0-4}
\begin{split}
\left\{
\begin{array}{ll}
&\bar{f}_{k,1}(\boldsymbol{\eta},\boldsymbol{\mu})=\bar{f}_{k,3}(\boldsymbol{\eta},\boldsymbol{\mu})
=\cdots=\bar{f}_{k,n}(\boldsymbol{\eta},\boldsymbol{\mu})=0, \\
&R>0,\quad \bar{D}_{k}(\boldsymbol{\eta},\boldsymbol{\mu})\neq0, \quad b\neq0
\end{array}
\right.
\end{split}
\end{equation}
has exactly $\ell$ distinct real solutions with respective to the variables $\boldsymbol{\eta}$. Here $\bar{f}_{k,j}(\boldsymbol{\eta},\boldsymbol{\mu})$ $(j=1,3,\ldots,n)$, and $\bar{D}_{k}(\boldsymbol{\eta},\boldsymbol{\mu})$ are respectively the numerator of the functions $f_{k,j}(\boldsymbol{\eta})$ and $D_k(\boldsymbol{\eta})$, with $\boldsymbol{\mu}=(\mu_1,\ldots,\mu_p)$ are parameters appearing in the averaged functions.
\end{theorem}

\begin{proof}
Assume that system \eqref{equ3-0-4} has $\ell$ distinct real solutions, which can be written as $\boldsymbol{\alpha}_1=(R^{(1)},X_3^{(1)},\ldots,X_n^{(1)})$, $\boldsymbol{\alpha}_2=(R^{(2)},X_3^{(2)},\ldots,X_n^{(2)}), \ldots, \boldsymbol{\alpha}_{\ell}=(R^{(\ell)},X_3^{(\ell)},\ldots,X_n^{(\ell)})$. Note that, for each solution $\boldsymbol{\alpha}_j$, $\bar{D}_{k}(\boldsymbol{\alpha}_j)\neq0$ implies that $D_k(\boldsymbol{\eta})\neq0$ in a suitable small neighborhood of $\boldsymbol{\eta}=(0,0,\ldots,0)$ with $R>0$. From Theorem \ref{averaging-thm}, we know that for $|\varepsilon|>0$ sufficiently small, system \eqref{equ-2} has a $T$-periodic solution $\boldsymbol{x}(t,\boldsymbol{\alpha}_j,\varepsilon)$ such that $\boldsymbol{x}(0,\boldsymbol{\alpha}_j,\varepsilon)\rightarrow\boldsymbol{\alpha}_j$ when $\varepsilon\rightarrow0$. This completes the proof.
\end{proof}

\begin{remark}\label{rem-semi-1}
Theorem \ref{semi-averaging} provides an effective and straightforward computation method to verify whether an obtained bound ($\ell$) for the number of limit cycles of a given differential system can be reached. Its main task is to find conditions on the parameters $\boldsymbol{\mu}$ for system \eqref{equ3-0-4} to have exactly $\ell$ distinct real solutions. In next section, we will give a systematic approach for solving semi-algebraic systems and analyzing zero-Hopf bifurcation of limit cycles by using symbolic computation methods automatically. Remark that, similar to system \eqref{equ3-0-4}, the semi-algebraic system for system \eqref{equ-2} to have a prescribed number of stable limit cycles bifurcating from the origin may also be formulated by using Routh--Hurwitz's stability criterion.
\end{remark}

Next we recall, from sparse elimination theory, the BKK bound for Bernstein, Khovanskii, and Kushnirenko \cite{Ber75,Kus76,Kho77}. The BKK bound counts the number of common complex solutions of a generic polynomial system. As a result, the BKK bound of the polynomial system $\{\bar{f}_{k,1}(\boldsymbol{\eta},\boldsymbol{\mu}),\bar{f}_{k,3}(\boldsymbol{\eta},\boldsymbol{\mu})
,\ldots,$ $\bar{f}_{k,n}(\boldsymbol{\eta},\boldsymbol{\mu})\}$ provides a bound for the maximum number of limit cycles that can bifurcate from the origin up to the $k$th-order averaging. More detailed information on the Bernstein's theorem and an expository account of recent work in this area can be found in \cite[Section 7]{DJD2004}. In Table \ref{Tab-1}, we present some values of the BKK bound of $H_2(n,N)$. The results on the mixed volume were computed by using Emiris and Canny's algorithm \cite{EC95}, for which the software is available at https://github.com/iemiris/MixedVolume-SparseResultants. There are two other software packages, {\sf MixedVol-2.0} \cite{LeeL11} and {\sf PHCpack} in Macaulay2 \cite{GPV13}, which can be used for computing mixed volumes of polynomial systems.

\begin{table}[h]
\begin{center}
\caption{Some values of the BKK bounds of $H_2(n,N)$.}\label{Tab-1}
\begin{tabular}{|c|c|c||c||c||c||c||c|}
  \hline
  \multicolumn{8}{|c|}{$n$}  \\
  \hline \hline
  & & 3 & 4 & 5 & 6 & 7 & $\cdots$\\
  \hline \hline
  \multirow{7}{*}{$N$}
  & 2 & 3 &  9 & 27 & 81 & 243 & \\
  & 3 & 3 &  9 & 27 & 81 & 243 & \\
  & 4 & 3 &  9 & 27 & 81 & 243 & \\
  & 5 & 3 & 9 & 27 & 81 & 243 & \\
  & 6 & 3 & 9 & 27 &  81 & 243 & \\
  & 7 & 3 & 9 & 27 & 81 & 243 & \\
 % & 8 & 3 & 9 & 27 & 81 & - & - &  \\
  & $\vdots$ &  &  &  &  &  & $\ddots$\\
  \hline
\end{tabular}
\end{center}
\end{table}

\begin{remark}\label{rem-bkk}
By observing the values of the BKK bounds of $H_2(n,N)$, we conjecture that
the BKK bound of the resulting polynomial system obtained from the averaged functions is independent of $N$. Moreover, we believe that the values of the BKK bound $=H_2(n,N)=3^{n-2}$. In Section \ref{sect4.1}, we will provide a concrete example to show that $H_2(3,3)=3$ by using our algorithmic approach. Similar results may also be established for other specific values of $n$ and $N$. How to provide rigorous proof of ``the BKK bound $=H_2(n,N)=3^{n-2}$'' for general $n$ and $N$ is an interesting question that remains for further investigation.
\end{remark}

\section{Algorithmic Analysis of Limit-Cycle Bifurcation}\label{sect3}

The process of using the averaging method for analyzing limit cycles of a given differential system can be divided into three steps \cite[Section 4]{hy19}. A minor modified version is stated as follows.
\begin{description}
  \item[Step 1:] Write a perturbed system of the form \eqref{equ-2} in the standard form of averaging \eqref{equ2-1} up to the $k$th-order in $\varepsilon$.
  \item[Step 2:] (a) Compute the exact formula of the $k$th-order integral function $\boldsymbol{y}_k(t,\boldsymbol{z})$ in \eqref{equ2-6}; (b) Derive the symbolic expression of the $k$th-order averaged function $\boldsymbol{f}_k(\boldsymbol{z})$ by \eqref{equ2-3}.
  \item[Step 3:] Determine the number $H_k(n,N)$ (i.e., the exact upper bound for the number of real isolated solutions of $\boldsymbol{f}_k(\boldsymbol{z})$).

\end{description}

We develop a symbolic program using Maple to implement the first two steps. The Maple program is mainly based on Theorem \ref{thm2-a}, $L$-multilinear map \eqref{equ2-2}, and Lemma \ref{lem-main-1}. The outline of the Maple program is presented in Section \ref{sect3.1}. For {\bf Step 3}, we first compute the mixed volume of a polynomial system (derived from the Maple program) to obtain an upper bound for the number of limit cycles, then we use Theorem \ref{semi-averaging} to check whether the obtained bound is reached. An algorithmic derivation for Theorem \ref{semi-averaging} is given in Section \ref{sect3.2}.

\subsection{Outline of Symbolic Maple Program}\label{sect3.1}

In this subsection, for the convenience of readers, we list the outline of symbolic Maple program developed in this paper, which can be used for computing the higher-order averaged functions of nonlinear differential systems. The program has a modular structure, and is formed by a main process together with two auxiliary functions. The source code of the Maple program is available at{\small{ {\textcolor{blue}{\underline{\url{https://github.com/Bo-Math/zero-Hopf}}}}}}.

The first auxiliary function, \textsc{StandardForm}$(S,k)$, is a direct implementation of the formula derivation in the proof of Lemma \ref{lem-main-1}, where $S=[\dot{x}_1,\dot{x}_2,\ldots,\dot{x}_n]$ denotes a differential system of the form \eqref{equ-2}. This function transforms a given differential system denoted by $S$ into the standard form of averaging $\emph{SF}=[dR/d\theta,dX_3/d\theta,\ldots,dX_n/d\theta]$.

The second auxiliary function, \textsc{OrderKFormula}$(k,n)$, computes formula of the $k$th-order integral function $\boldsymbol{y}_k(\theta,\boldsymbol{z})$ in \eqref{equ2-6}. Correctness of it follows from Theorem \ref{thm2-a}. We deduce explicitly the formulae of $\boldsymbol{y}_k$'s up to $k=3$ in Appendix \ref{sect-Algo}. The times of computation using the Maple program for some examples are given in Table \ref{Tab-A}.

The main process of the program, \textsc{AveragedFunctions}$(\emph{SF},k)$, is based on the functions \textsc{StandardForm}$(S,k)$ and \textsc{OrderKFormula}$(k,n)$, which provides a straightforward calculation method to derive the exact expression of the $k$th-order averaged functions for a given differential system. We remark that the program requires quite heavy computation, which grows notably when one of the averaging order $k$, degree $N$, or dimension $n$ of the required systems increases. Because of this, each time we compute the $k$th-order averaged function, we will update the obtained standard form $\emph{SF}$ by using the conditions of $\boldsymbol{f}_1\equiv\boldsymbol{f}_2\equiv\cdots\boldsymbol{f}_{k-1}=0$.

In Section \ref{sect4}, we will present several examples to demonstrate the applicability and the computational efficiency of the Maple program.

\subsection{Algorithmic Derivation for Theorem \ref{semi-averaging}}\label{sect3.2}

Our purpose is to derive sufficient conditions on the parameters for a given differential system of the form \eqref{equ-2} to have a prescribed number of limit cycles bifurcating from the origin. In the following we propose a general algorithmic approach for automatically analyzing Theorem \ref{semi-averaging} by using methods from symbolic computation. This approach is based on the one for solving semi-algebraic systems proposed by Wang, Xia and Niu \cite{WDMXBC2005,niuwang08}. The main steps of our computational approach are summarized as follows.

\smallskip

{\bf STEP 1}. Based on the Maple program developed in Section
\ref{sect3.1}, formulate the semi-algebraic system \eqref{equ3-0-4} from a differential system of the form \eqref{equ-2}. Denote by $\mathcal{S}$ the semi-algebraic system for solving, $\Gamma$ the set of inequalities of $\mathcal{S}$, $\mathcal{F}$ the set of polynomials in $\Gamma$, and $\mathcal{P}$ the set of polynomials in the equations of $\mathcal{S}$.

\smallskip

{\bf STEP 2}. Triangularize the set $\mathcal{P}$ of
polynomials to obtain one or several (regular) triangular sets
$\mathcal{T}_{k}$ by using the method of triangular decomposition or
Gr\"obner bases.

\smallskip

{\bf STEP 3}. For each triangular set $\mathcal{T}_{k}$, use the
polynomial set $\mathcal{F}$ to compute an algebraic variety $V$ in
$\boldsymbol{\mu}$ by means of real solution classification (e.g.,
Yang--Xia's method \cite{YHX01,LYBX05} or Lazard--Rouillier's method
\cite{DLFR}), which decomposes the parameter space $\mathbb{R}^{p}$ into finitely many cells such that in each cell the number of real zeros of $\mathcal{T}_{k}$ and the signs of polynomials in $\mathcal{F}$ at these real zeros remain invariant. The algebraic variety is defined by polynomials in $\boldsymbol{\mu}$. Then take a rational sample point from each cell by using the method of PCAD or critical points \cite{msed2007}, and isolate the real zeros of $\mathcal{T}_{k}$ by rational intervals at this sample point. In this way, the number of real zeros of $\mathcal{T}_{k}$ and the signs of polynomials in $\mathcal{F}$ at these real zeros in each cell are determined.

\smallskip

{\bf STEP 4}. Determine the signs of (the factors of) the defining polynomials of $V$ at each sample point. Formulate the conditions on $\boldsymbol{\mu}$ according to the signs of these defining polynomials at the sample points in those cells in which the system $S$ has exactly the number of real solutions we want.

\smallskip

{\bf STEP 5}. Output the conditions on the parameter $\boldsymbol{\mu}$ such that the differential system has a prescribed number of limit cycles bifurcating from
the origin.

\smallskip

There are several packages or software for realization of certain steps in our approach. For example, the method of discriminant varieties of Lazard and Rouillier \cite{DLFR} (implemented as a Maple package {\sf DV} by Moroz and Rouillier), and the Maple package {\sf DISCOVERER} (see also recent improvements in the Maple package RegularChains[SemiAlgebraicSetTools]), developed by Xia, implements the methods of Yang and Xia \cite{LYBX05} for real solution classification. In Section \ref{sect4}, we will apply our general algorithmic approach to analyze zero-bifurcations for several concrete differential systems in order to show its feasibility.

\section{Experiments}\label{sect4}
In this section, we explain how to apply the algorithmic tests to the study of zero-Hopf bifurcations of polynomial differential systems and illustrate some of the computational steps by a famous jerk differential system. In addition, using our computational approach, we also analyze the zero-Hopf bifurcation of limit cycles for a class of generalized Lorenz systems and a 4D hyperchaotic differential system. The experimental results show the applicability and efficiency of our algorithmic approach. All the experiments were made in Maple 17 on a Windows 10 laptop with 4 CPUs 2.9GHz and 8192M RAM.

\subsection{Illustrative Example}\label{sect4.1}
We study the zero-Hopf bifurcation of the 3D jerk system:
\begin{equation}\label{equ3-7-1}
\begin{split}
\dot{x}&=y,\quad \dot{y}=z,\\
\dot{z}&=-az-bx+cy+xy^2-x^3,
\end{split}
\end{equation}
where $a,b,c\in\mathbb{R}$. System \eqref{equ3-7-1} corresponds to a nonlinear third-order differential equation studied by Vaidyanathan \cite{SVai2017}, showing that this equation can exhibit a rich range of dynamical behavior. As shown
by \cite[Proposition 1]{FBAM21} that, the origin is a zero-Hopf equilibrium when $a=b=0$ and $c<0$. Here, using the second-order averaging method, we restudy the limit cycles that can bifurcate from the origin of the jerk system \eqref{equ3-7-1}. To do this, consider the vector $(a,b,c)$ given by
\begin{equation}\label{equ3-7-1-1}
\begin{split}
a&=\varepsilon a_1+\varepsilon^2 a_2,\quad
b=\varepsilon b_1+\varepsilon^2 b_2,\\
c&=-\beta^2+\varepsilon c_1+\varepsilon^2 c_2,\quad \beta\neq0,
\end{split}
\end{equation}
where the constants $a_i$, $b_i$ and $c_i$ are all real parameters. Then the jerk system becomes
\begin{equation}\label{equ3-7-2}
\begin{split}
\dot{x}&=y,\quad \dot{y}=z,\\
\dot{z}&=-(\varepsilon a_1+\varepsilon^2 a_2)z-(\varepsilon b_1+\varepsilon^2b_2)x\\
&\quad+(-\beta^2+\varepsilon c_1+\varepsilon^2 c_2)y+xy^2-x^3.
\end{split}
\end{equation}

Our main result on the limit cycles of system \eqref{equ3-7-2} is the following. Its proof can be found in Appendix \ref{sect-chun}.

\begin{theorem}\label{main-2}
The following statements hold for $|\varepsilon|>0$ sufficiently small.
\begin{itemize}
\item [(i)] The first-order averaging does not provide any
information about limit cycles that bifurcate from the origin.
\item [(ii)] System \eqref{equ3-7-2} has, up to the second-order averaging, at most 3 limit cycles bifurcating from the origin, and this number can be reached if one of the following 2 conditions holds:
\begin{equation}\label{equ3-7-3}
\begin{split}
\mathcal{C}_0&=[R_1<0,\,R_2<0,\,0<R_3,\,0<R_4]\wedge\bar{\mathcal{C}},\\ \mathcal{C}_1&=[0<R_1,\,0<R_2,\,0<R_3,\,R_4<0]\wedge\bar{\mathcal{C}},
\end{split}
\end{equation}
where
\begin{equation}\label{equ3-7-3-1}
\begin{split}
R_1&=\beta^2-3,\quad R_2=\beta^2a_2+2b_2,\\
R_3&=2\beta^2a_2-b_2,\quad R_4=\beta^2a_2-b_2,\\
\bar{\mathcal{C}}&=[\beta\neq0,R_1\neq0,R_2\neq0,R_3\neq0,R_4\neq0].\nonumber
\end{split}
\end{equation}

\end{itemize}
\end{theorem}

\begin{remark}\label{remmark-jerk}
Theorem \ref{main-2} is consistent with statement (1) of \cite[Theorem 2]{FBAM21}. The condition $\bar{\mathcal{C}}$ is of type \textit{border polynomial} excluding some exceptional parameter values, provided that $\bar{\mathcal{C}}$ does not vanish. We prove Theorem \ref{main-2} by using the second-order averaging method. Since the second-order averaged functions cannot be identically zero, it follows that the higher-order of averaging cannot be applied to find eventually more limit cycles.
\end{remark}

\subsection{Other Models and Remarks}\label{sect4.2}

In order to save space, the details of our results on the zero-Hopf bifurcations for a class of generalized Lorenz systems and a four-dimensional hyperchaotic differential system are placed in Appendix \ref{sect-ex}.

\section{Conclusion}\label{sect5}
We develop a symbolic Maple program for computing the averaged functions of any order for nonlinear differential systems. With the aid of this program, we reduce the analysis of zero-Hopf bifurcations to the solution of a semi-algebraic system and introduce a systematical computational approach for rigorously analyzing the conditions on the parameters under which a considered differential system has a prescribed number of limit cycles bifurcating from a zero-Hopf equilibrium. The results of experiments we performed verify the effectiveness of the proposed approach.

It would be interesting to employ our algorithmic approach for analyzing zero-Hopf bifurcations of high-dimensional polynomial differential systems in many different fields, which are of high interest in nature sciences and engineering. Our program, originally developed for computing averaged functions of continuous differential systems, can also be generalized to study limit cycles for discontinuous differential systems. As our further work, it is of great interest to give rigorous proof of ``the BKK bound $=H_2(n,N)=3^{n-2}$'' for general $n$ and $N$ (see Remark \ref{rem-bkk}). Moreover, how to provide a good bound for the number of limit cycles of system \eqref{equ-2} up to $k$th-order averaging ($k\geq2$) is also worthy of further study.

\section*{Acknowledgments}
Huang's work is partially supported by the National Natural Science Foundation of China (NSFC 12101032 and NSFC 12131004). The author wishes to thank Professor Dongming Wang and Professor Deren Han for their support and encouragement.

\bibliographystyle{plain}
\bibliography{ref}

\begin{thebibliography}{10}

\bibitem{adrt2016}
Valery Antonov, Diana Doli\'canin, Valery~G. Romanovski, and J\'anos T\'oth.
\newblock Invariant planes and periodic oscillations in the may--leonard
  asymmetric model.
\newblock {\em MATCH Commun. Math. Comput. Chem.}, 76:455--474, 2016.

\bibitem{BMS06}
Inmaculada Baldom\'a and Tere~M. Seara.
\newblock Brakdown of heteroclinic orbits for some analytic unfoldings of the
  hopf–zero singulairty.
\newblock {\em J. Nonlinear Sci.}, 16:543--582, 2006.

\bibitem{BLV18}
Luis Barreira, Jaume Llibre, and Claudia Valls.
\newblock Limit cycles bifurcating from a zero-hopf singularity in arbitrary
  dimension.
\newblock {\em Nonlinear Dyn.}, 92:1159--1166, 2018.

\bibitem{Ber75}
David~N. Bernstein.
\newblock The number of roots of a system of equations.
\newblock {\em Functional Anal. Appl.}, 9:1--4, 1975.

\bibitem{BiYu99}
Qin~S. Bi and Pei Yu.
\newblock Symbolic computation of normal forms for semi-simple cases.
\newblock {\em J. Comput. Appl. Math.}, 102:195--220, 1999.

\bibitem{bhlr2015}
Fran\c{c}ois Boulier, Mao~A. Han, Fran\c{c}ois Lemaire, and Valery~G.
  Romanovski.
\newblock Qualitative investigation of a gene model using computer algebra
  algorithms.
\newblock {\em Program. Comput. Soft.}, 41:105--111, 2015.

\bibitem{FBAM21}
Francisco Braun and Ana~C. Mereu.
\newblock Zero-hopf bifurcation in a 3{D} jerk system.
\newblock {\em Nonlinear Anal. RWA}, 59:103245--1--8, 2021.

\bibitem{BB85}
B.~Buchberger.
\newblock Gr{\"o}bner bases: An algorithmic method in polynomial ideal theory.
\newblock pages 184--232, Reidel, Dordrecht, 1985. In: Multidimensional Systems
  Theory (N.K. Bose, ed.).

\bibitem{BLM2016}
Claudio Buzzi, Jaume Llibre, and Jo$\tilde{\mbox{a}}$o Medrado.
\newblock Hopf and zero-{H}opf bifurcations in the {H}indmarsh-{R}ose system.
\newblock {\em Nonlinear Dyn.}, 83:1549--1556, 2016.

\bibitem{CCMYZ13}
Chang~B. Chen, Robert~M. Corless, Marc~M. Maza, Pei Yu, and Yi~M. Zhang.
\newblock An application of regular chain theory to the study of limit cycles.
\newblock {\em Int. J. Bifur. Chaos}, 23:1350154--1--21, 2013.

\bibitem{GEC75}
George~E. Collins.
\newblock Quantifier elimination for real closed fields by cylindrical
  algebraic decomposition.
\newblock pages 134--183, Springer, Berlin Heidelberg, 1975. In: Proceedings of
  the second GI Conference on Adtomata Theory and Formal Languages (H.
  Barkhage, ed.).

\bibitem{CHH91}
George~E. Collins and Hoon Hong.
\newblock Partial cylindrical algebraic decomposition for quantifier
  elimination.
\newblock {\em J. Symb. Comput.}, 12:299--328, 1991.

\bibitem{DJD2004}
David~A. Cox, John Little, and Donal O'Shea.
\newblock {\em Using Algebraic Geometry}.
\newblock Springer, New York, 2004.

\bibitem{msed2007}
Mohab Safey~El Din.
\newblock Testing sign conditions on a multivariate polynomial and
  applications.
\newblock {\em Math. Comput. Sci.}, 1:177--207, 2007.

\bibitem{EC95}
Ioannis~Z. Emiris and John~F. Canny.
\newblock Efficient incremental algorithms for the sparse resultant and the
  mixed volume.
\newblock {\em J. Symb. Comput.}, 20:117--149, 1995.

\bibitem{EJL17}
Rodrigo~D. Euz\'ebio and Jaume Llibre.
\newblock Zero-hopf bifurcation in a {C}hua system.
\newblock {\em Nonlinear Anal. RWA}, 37:31--40, 2017.

\bibitem{GPV13}
Elizabeth Gross, Sonja Petrovi\'c, and Jan Verschelde.
\newblock Interfacing with phcpack.
\newblock {\em The Journal of Software for Algebra and Geometry}, 5:20--25,
  2013.

\bibitem{GH93}
John Guckenheimer and Philip Holmes.
\newblock {\em Nonlinear {O}scillations, {D}ynamical {S}ystems, and
  {B}ifurcations of {V}ector {F}ields}.
\newblock Springer, New York, 1993.

\bibitem{HGR10}
Mao~A. Han and Valery~G. Romanovski.
\newblock Estimating the number of limit cycles in polynomials systems.
\newblock {\em J. Math. Anal. Appl.}, 368:491--497, 2010.

\bibitem{HY12}
Mao~A. Han and Pei Yu.
\newblock {\em Normal {F}orms, {M}elnikov {F}unctions, and {B}ifurcations of
  {L}imit {C}ycles}.
\newblock Springer, New York, 2012.

\bibitem{HLS97}
Hoon Hong, Richard Liska, and Stanly Steinberg.
\newblock Testing stability by quantifier elimination.
\newblock {\em J. Symb. Comput.}, 24:161--187, 1997.

\bibitem{HTX15}
Hoon Hong, Xiao~X. Tang, and Bi~C. Xia.
\newblock Special algorithm for stability analysis of multistable biological
  regulatory systems.
\newblock {\em J. Symb. Comput.}, 70:112--135, 2015.

\bibitem{HNW2022}
Bo~Huang, Wei Niu, and Dong~M. Wang.
\newblock Symbolic computation for the qualitative theory of differential
  equations.
\newblock {\em Acta Math. Sci.}, 42B:2478--2504, 2022.

\bibitem{hy19}
Bo~Huang and Chee~K. Yap.
\newblock An algorithmic approach to small limit cycles of nonlinear
  differential systems: The averaging method revisited.
\newblock {\em J. Symb. Comput.}, 115:492--517, 2023.

\bibitem{Kho77}
Askold~G. Khovanskii.
\newblock Newton polytopes and toric varieties.
\newblock {\em Functional Anal. Appl.}, 11:289--298, 1977.

\bibitem{Kus76}
Anatoli~G. Kouchnirenko.
\newblock Newton polytopes and the b\'ezout theorem.
\newblock {\em Functional Anal. Appl.}, 10:233--235, 1976.

\bibitem{YK04}
Yury Kuznetsov.
\newblock {\em Elements of Applied Bifurcation Theory}.
\newblock Springer-Verlag, New York, 2004.

\bibitem{DLFR}
Daniel Lazard and Fabrice Rouillier.
\newblock Solving parametric polynomial systems.
\newblock {\em J. Symb. Comput.}, 42:636--667, 2007.

\bibitem{LeeL11}
Tsung-Lin Lee and Tien-Yien Li.
\newblock Mixed volume computation in solving polynomial systems.
\newblock {\em Contemporary Mathematics}, 556:97--112, 2011.

\bibitem{jlmrc2018}
Jaume Llibre and Murilo~R. C\^{a}ndido.
\newblock Zero-hopf bifurcations in a hyperchaotic lorenz system ii.
\newblock {\em Int. J. Nonlinear Sci.}, 25:3--26, 2018.

\bibitem{JLAM2016}
Jaume Llibre and Amar Makhlouf.
\newblock Zero-{H}opf bifurcation in the generalized {M}ichelson system.
\newblock {\em Chaos, Solitons and Fractals}, 89:228--231, 2016.

\bibitem{JA20}
Jaume Llibre and Ammar Makhlouf.
\newblock Zero-hopf periodic orbits for a r{\"o}ssler differential system.
\newblock {\em Int. J. Bifurcation and Chaos.}, 30:2050170--1--6, 2020.

\bibitem{LMB09}
Jaume Llibre, Ammar Makhlouf, and Sabrina Badi.
\newblock 3-dimensional hopf bifurcation via averaging theory of second order.
\newblock {\em Discret. Contin. Dyn. Syst.}, 25:1287--1295, 2009.

\bibitem{LMS15}
Jaume Llibre, Richard Moeckel, and Carles Sim\'o.
\newblock {\em Central Configurations, Periodic Orbits, and Hamiltonian
  Systems}.
\newblock Birkh{\"a}user, Basel, 2015.

\bibitem{LNT14}
Jaume Llibre, Douglas~D. Novaes, and Marco~A. Teixeira.
\newblock Higher order averaging theory for finding periodic solutions via
  {B}rouwer degree.
\newblock {\em Nonlinearity}, 27:563--583, 2014.

\bibitem{jlzy2021}
Jaume Llibre and Yu~Z. Tian.
\newblock The zero-hopf bifurcations of a four-dimensional hyperchaotic system.
\newblock {\em J. Math. Phys.}, 62:052703--1--9, 2021.

\bibitem{jlcv2011}
Jaume Llibre and Claudia Valls.
\newblock Hopf bifurcation for some analytic differential systems in
  $\mathbb{R}^3$ via averaging theory.
\newblock {\em Discret. Contin. Dyn. Syst.}, 30:779--790, 2011.

\bibitem{JLXZ09}
Jaume Llibre and Xiang Zhang.
\newblock Hopf bifurcation in higher dimensional differential systems via the
  averaging method.
\newblock {\em Pac. J. Math.}, 240:321--341, 2009.

\bibitem{lzhang2011}
Jaume Llibre and Xiang Zhang.
\newblock On the hopf-zero bifurcation of the michelson system.
\newblock {\em Nonlinear Anal. RWA}, 12:1650--1653, 2011.

\bibitem{crla2018}
Cristian L\v{a}zureanu.
\newblock Integrable deformations of three-dimensional chaotic systems.
\newblock {\em Int. J. Bifur. Chaos}, 28:1850066--1--7, 2018.

\bibitem{acj17}
Adam Mahdi, Claudio Pessoa, and Jonathan~D. Hauenstein.
\newblock A hybrid symbolic-numerical approach to the center-focus problem.
\newblock {\em J. Symb. Comput.}, 82:57--73, 2017.

\bibitem{ykm1993}
Yiu-Kwong Man.
\newblock Computing closed form solutions of first order odes using the
  prelle--singer procedure.
\newblock {\em J. Symb. Comput.}, 16:423--443, 1993.

\bibitem{niuwang08}
Wei Niu and Dong~M. Wang.
\newblock Algebraic approaches to stability analysis of biological systems.
\newblock {\em Math. Comput. Sci.}, 1:507--539, 2008.

\bibitem{dn17}
Douglas~D. Novaes.
\newblock {\em An Equivalent Formulation of the Averaged Functions via Bell
  Polynomials}, pages 141--145.
\newblock Trends in Mathematics 8. Springer, New York, 2017.

\bibitem{dpxz2009}
Ding~H. Pi and Xiang Zhang.
\newblock Limit cycles of differential systems via the averaging method.
\newblock {\em Canad. Appl. Math. Quart.}, 7:243--269, 2009.

\bibitem{vd09}
Valery~G. Romanovski and Douglas~S. Shafer.
\newblock {\em The Center and Cyclicity Problems: A Computational Algebra
  Approach}.
\newblock Birkh{\"a}user, Boston, 2009.

\bibitem{SVM2007}
Jan~A. Sanders, Ferdinand Verhulst, and James Murdock.
\newblock {\em Averaging Methods in Nonlinear Dynamical Systems}.
\newblock Applied Mathematical Sciences 59. Springer, New York, 2007.

\bibitem{SJM84}
J{\"u}gren Scheurle and Jerrold Marsden.
\newblock Bifurcation to quasi-periodic tori in the interaction of steady state
  and hopf bifurcations.
\newblock {\em SIAM. J. Math. Anal.}, 15:1055--1074, 1984.

\bibitem{IRS74}
Igor~R. Shafarevich.
\newblock {\em Basic Algebraic Geometry}.
\newblock Springer, Berlin, 1974.

\bibitem{SH2017}
Xian~B. Sun and Wen~T. Huang.
\newblock Bounding the number of limit cycles for a polynomial {L}i\'enard
  system by using regular chains.
\newblock {\em J. Symb. Comput.}, 79:197--210, 2017.

\bibitem{TiYu14}
Yun Tian and Pei Yu.
\newblock An explicit recursive formula for computing the normal forms
  associated with semisimple cases.
\newblock {\em Commun. Nonlinear Sci. Numer. Simul.}, 19:2294--2308, 2014.

\bibitem{SVai2017}
Sundarapandian Vaidyanathan.
\newblock A new 3-{D} jerk chaotic system with two cubic nonlinearities and its
  adaptive backstepping control.
\newblock {\em Arch. Control Sci.}, 27:409--439, 2017.

\bibitem{dw91}
Dong~M. Wang.
\newblock Mechanical manipulation for a class of differential systems.
\newblock {\em J. Symb. Comput.}, 12:233--254, 1991.

\bibitem{DW01}
Dong~M. Wang.
\newblock {\em Elimination Methods}.
\newblock Springer-Verlag, Wien, 2001.

\bibitem{WDMXBC2005}
Dong~M. Wang and Bi~C. Xia.
\newblock Stability analysis of biological systems with real solution
  classification.
\newblock pages 354--361, Beijing, China, 2005. In: Proceedings of the 2005
  International Symposium on Symbolic and Algebraic Computation.

\bibitem{WTW00}
Wen-Ts{\"u}n Wu.
\newblock {\em Mathematics Mechanization}.
\newblock Science Press/Kluwer Academic, Beijing, 2000.

\bibitem{YHX01}
Lu~Yang, Xiao~R. Hou, and Bi~C. Xia.
\newblock A complete algorithm for automated discovering of a class of
  inequality-type theorems.
\newblock {\em Sci. China (Ser. F)}, 44:33--49, 2001.

\bibitem{LYBX05}
Lu~Yang and Bi~C. Xia.
\newblock Real solution classifications of parametric semi-algebraic systems.
\newblock pages 281--289, Herstellung und Verlag, Norderstedt, 2005. In:
  Algorithmic Algebra and Logic -- Proceedings of the A3L 2005 (A. Dolzmann, A.
  Seidl, and T. Sturm, eds.).

\bibitem{BPY20}
Bing Zeng and Pei Yu.
\newblock Analysis of zero-{H}opf bifurcation in two {R}{\"o}ssler systems
  using normal form theory.
\newblock {\em Int. J. Bifur. Chaos}, 30:2030050--1--14, 2020.

\bibitem{zyxc2019}
Cheng~Q. Zhou, Chun~H. Yang, De~G. Xu, and Chao~Y. Chen.
\newblock Dynamic analysis and finite-time synchronization of a new
  hyperchaotic system with coexisting attractors.
\newblock {\em IEEE Access}, 7:52896--52902, 2019.

\end{thebibliography}

\appendix
\newpage

\section{Proof of Lemma \ref{lem-main-1}}\label{sec-A}

To apply the averaging method, we rescale the variables by setting
\begin{equation}\label{equA3-1}
\begin{split}
(x_1,x_2,\ldots,x_n)=\left(\varepsilon X_1,\varepsilon X_2,\ldots,\varepsilon X_n\right).
\end{split}
\end{equation}
Then system \eqref{equ-2} becomes
\begin{equation}\label{equA3-2}
\begin{split}
\dot{X}_1&=-bX_2+\sum_{m\geq2}^{N}\varepsilon^{m-1}\sum_{i_1+\cdots+i_n=m}a_{i_1,\ldots,i_n}X_1^{i_1}\cdots X_n^{i_n}\\
&+\sum_{j=1}^k\varepsilon^j\sum_{m^*\geq1}^{N}\varepsilon^{m^*-1}\sum_{i_1+\cdots+i_n=m^*}a_{j,i_1,\ldots,i_n}X_1^{i_1}\cdots X_n^{i_n},\\
\dot{X}_2&=bX_1+\sum_{m\geq2}^{N}\varepsilon^{m-1}\sum_{i_1+\cdots+i_n=m}b_{i_1,\ldots,i_n}X_1^{i_1}\cdots X_n^{i_n}\\
&+\sum_{j=1}^k\varepsilon^j\sum_{m^*\geq1}^{N}\varepsilon^{m^*-1}\sum_{i_1+\cdots+i_n=m^*}b_{j,i_1,\ldots,i_n}X_1^{i_1}\cdots X_n^{i_n},\\
\dot{X}_s&=\sum_{m\geq2}^{N}\varepsilon^{m-1}\sum_{i_1+\cdots+i_n=m}c_{i_1,\ldots,i_n,s}X_1^{i_1}\cdots X_n^{i_n}\\
&+\sum_{j=1}^k\varepsilon^j\sum_{m^*\geq1}^{N}\varepsilon^{m^*-1}\sum_{i_1+\cdots+i_n=m^*}c_{j,i_1,\ldots,i_n,s}X_1^{i_1}\cdots X_n^{i_n}.
\end{split}
\end{equation}
Making the change of variables
\begin{equation}\label{equA3-3}
\begin{split}
X_1=R\cos\theta,\quad X_2=R\sin\theta,\quad X_s=X_s,\quad s=3,\ldots,n\nonumber
\end{split}
\end{equation}
with $R>0$, system \eqref{equA3-2} becomes
\begin{equation}\label{equA3-4}
\begin{split}
\frac{dR}{dt}&=\cos\theta\cdot S_{1}(\theta,R,X_3,\ldots,X_n,\varepsilon)+\sin\theta\cdot S_{2}(\theta,R,X_3,\ldots,X_n,\varepsilon),\\
\frac{d\theta}{dt}&=b+\frac{\cos\theta}{R}\cdot S_{2}(\theta,R,X_3,\ldots,X_n,\varepsilon)\\
&\quad-\frac{\sin\theta}{R}\cdot S_{1}(\theta,R,X_3,\ldots,X_n,\varepsilon),\\
\frac{dX_s}{dt}&=S_{s}(\theta,R,X_3,\ldots,X_n,\varepsilon),\quad s=3,\ldots,n,
\end{split}
\end{equation}
where
\begin{align*}
&S_{1}=\sum_{m\geq2}^N\varepsilon^{m-1}\sum_{i_1+\cdots+i_n=m}a_{i_1,\ldots,i_n}\Psi(\theta,R,X_3,\ldots,X_n)\\
&+\sum_{j=1}^k\varepsilon^j\sum_{m^*\geq1}^N\varepsilon^{m^*-1}\sum_{i_1+\cdots+i_n=m^*}a_{j,i_1,\ldots,i_n}\Psi(\theta,R,X_3,\ldots,X_n),\\
&S_{2}=\sum_{m\geq2}^N\varepsilon^{m-1}\sum_{i_1+\cdots+i_n=m}b_{i_1,\ldots,i_n}\Psi(\theta,R,X_3,\ldots,X_n)\\
&+\sum_{j=1}^k\varepsilon^j\sum_{m^*\geq1}^N\varepsilon^{m^*-1}\sum_{i_1+\cdots+i_n=m^*}b_{j,i_1,\ldots,i_n}\Psi(\theta,R,X_3,\ldots,X_n),\\
&S_{s}=\sum_{m\geq2}^N\varepsilon^{m-1}\sum_{i_1+\cdots+i_n=m}c_{i_1,\ldots,i_n,s}\Psi(\theta,R,X_3,\ldots,X_n)\\
&+\sum_{j=1}^k\varepsilon^j\sum_{m^*\geq1}^N\varepsilon^{m^*-1}
\sum_{i_1+\cdots+i_n=m^*}c_{j,i_1,\ldots,i_n,s}\Psi(\theta,R,X_3,\ldots,X_n)
\end{align*}
with $\Psi(\theta,R,X_3,\ldots,X_n)=(R\cos\theta)^{i_1}(R\cos\theta)^{i_2}X_3^{i_3}\cdots X_n^{i_n}$.

Since $b\neq0$, one can easily verify that in a suitable small neighborhood of $(R,X_3,\ldots,X_n)=(0,0,\ldots,0)$ with $R>0$ we always have $d \theta/dt\neq0$. Then taking $\theta$ as the new independent variable, in a neighborhood of $(R,X_3,\ldots,X_n)=(0,0,\ldots,0)$, system \eqref{equA3-4} becomes
\begin{equation}\label{equA3-6}
\begin{split}
\frac{dR}{d\theta}&=\frac{R\big[\cos\theta\cdot S_{1}(\theta,R,X_3,\ldots,X_n,\varepsilon)+\sin\theta\cdot S_{2}(\theta,R,X_3,\ldots,X_n,\varepsilon)\big]}{bR+\cos\theta\cdot S_{2}(\theta,R,X_3,\ldots,X_n,\varepsilon)-\sin\theta\cdot S_{1}(\theta,R,X_3,\ldots,X_n,\varepsilon)},\\
\frac{dX_s}{d\theta}&=\frac{R\cdot S_{s}(\theta,R,X_3,\ldots,X_n,\varepsilon)}{bR+\cos\theta\cdot S_{2}(\theta,R,X_3,\ldots,X_n,\varepsilon)-\sin\theta\cdot S_{1}(\theta,R,X_3,\ldots,X_n,\varepsilon)}.
\end{split}
\end{equation}
By carrying Taylor expansion of expressions in \eqref{equA3-5} with respective to the variable $\varepsilon$ around $\varepsilon=0$, one obtains the functions $F_{j_1,j_2}(\theta,R,X_3,\ldots,X_n)$, $j_1=1,2,\ldots,k$, $j_2=1,3,\ldots,n$ in \eqref{equ3-0-1}. This ends the proof of Lemma \ref{lem-main-1}.

\section{Proof of Theorem \ref{main-th-1}}\label{A-B}

Fixing $k=1$ and applying Lemma \ref{lem-main-1} to system \eqref{equ-2}, we obtain the following differential system
\begin{equation}\label{equ3-1-5}
\begin{split}
\frac{dR}{d\theta}&=\varepsilon F_1(\theta,R,X_3,\ldots,X_n)+\mathcal{O}(\varepsilon^2),\\
\frac{dX_s}{d\theta}&=\varepsilon F_s(\theta,R,X_3,\ldots,X_n)+\mathcal{O}(\varepsilon^2),\quad s=3,\ldots,n,
\end{split}
\end{equation}
where
\begin{equation}\label{equ3-1-6}
\begin{split}
F_1&=\frac{1}{b}\Big(\sum_{i_1+\cdots+i_n=2}(a_{i_1,i_2,\ldots,i_n}\cos\theta+b_{i_1,i_2,\ldots,i_n}\sin\theta)\times
\Psi(\theta,R,X_3,\ldots,X_n)\\
&+\sum_{i_1+\cdots+i_n=1}(a_{1,i_1,i_2,\ldots,i_n}\cos\theta
+b_{1,i_1,i_2,\ldots,i_n}\sin\theta)\times
\Psi(\theta,R,X_3,\ldots,X_n)\Big),\\
F_s&=\frac{1}{b}\Big(\sum_{i_1+\cdots+i_n=2} c_{i_1,i_2,\ldots,i_n,s}\times\Psi(\theta,R,X_3,\ldots,X_n)\\
&\quad+\sum_{i_1+\cdots+i_n=1} c_{1,i_1,i_2,\ldots,i_n,s}\times\Psi(\theta,R,X_3,\ldots,X_n)\Big).\nonumber
\end{split}
\end{equation}
Here $\Psi(\theta,R,X_3,\ldots,X_n)=(R\cos\theta)^{i_1}(R\cos\theta)^{i_2}X_3^{i_3}\cdots X_n^{i_n}$. Now system \eqref{equ3-1-5} comes into the standard form \eqref{equ2-1} for applying the averaging method with $\boldsymbol{x}=\boldsymbol{\eta}=(R,X_3,\ldots,X_n)$, $t=\theta$, $T=2\pi$, and
\[\boldsymbol{F}(\theta,\boldsymbol{\eta})=(F_1(\theta,\boldsymbol{\eta}),F_3(\theta,\boldsymbol{\eta}),
\ldots,F_n(\theta,\boldsymbol{\eta})).\]
According to \eqref{equ2-3}, we need to compute the first-order averaged functions:
\begin{equation}\label{equ3-1-7}
\begin{split}
f_{1,s}(\boldsymbol{\eta})=\int_0^{2\pi}F_s(\theta,\boldsymbol{\eta})d\theta,\quad s=1,3,\ldots,n.
\end{split}
\end{equation}
After some calculations we obtain
\begin{small}
\begin{equation}\label{equ3-1-8}
\begin{split}
f_{1,1}(\boldsymbol{\eta})&=\frac{\pi R}{b}\Big(a_{1,1,0,\boldsymbol{0}_{n-2}}
+b_{1,0,1,\boldsymbol{0}_{n-2}}+\sum_{j=3}^n(a_{1,0,\boldsymbol{e}_j}
+b_{0,1,\boldsymbol{e}_j})X_j\Big),\\
f_{1,s}(\boldsymbol{\eta})&=\frac{\pi}{b}\Big((c_{2,0,\boldsymbol{0}_{n-2},s}
+c_{0,2,\boldsymbol{0}_{n-2},s})R^2+2\sum_{j=3}^nc_{1,0,0,\boldsymbol{e}_j,s}X_j\\
&\quad+2\sum_{3\leq j_1\leq j_2\leq n}c_{0,0,\boldsymbol{e}_{j_1j_2},s}X_{j_1}X_{j_2}
\Big),\quad s=3,\ldots,n,
\end{split}
\end{equation}
\end{small}
where $\boldsymbol{e}_{j}\in\mathbb{N}^{n-2}$ is the unit vector with $j$th entry equal to 1 ($\mathbb{N}$ denotes the set of all nonnegative integers), and $\boldsymbol{e}_{j_1j_2}\in\mathbb{N}^{n-2}$ has the sum of the $j_1$th and $j_2$th entries equal to $2$ and the other equal to 0.

Using the averaging theorem to study limit cycles of system \eqref{equ3-1-5}, we need to know the number of simple zeros of system \eqref{equ3-1-8}. So we should study the zeros of the algebraic equations
\begin{equation}\label{equ3-1-9}
\begin{split}
&a_{1,1,0,\boldsymbol{0}_{n-2}}
+b_{1,0,1,\boldsymbol{0}_{n-2}}+\sum_{j=3}^n(a_{1,0,\boldsymbol{e}_j}
+b_{0,1,\boldsymbol{e}_j})X_j=0,\\
&(c_{2,0,\boldsymbol{0}_{n-2},s}
+c_{0,2,\boldsymbol{0}_{n-2},s})R^2+2\sum_{3\leq j_1\leq j_2\leq n}c_{0,0,\boldsymbol{e}_{j_1j_2},s}X_{j_1}X_{j_2}\\&+2\sum_{j=3}^nc_{1,0,0,\boldsymbol{e}_j,s}X_j=0,
\end{split}
\end{equation}
for $s=3,\ldots,n$. Isolating $R$ from the equation in \eqref{equ3-1-9} for $s=3$, taking into account that $R>0$ and substituting it into the other equations of \eqref{equ3-1-9}, the numerator becomes a polynomial equation in $X_3,X_4,\ldots,X_n$ of degree 2 for $s=4,\ldots,n$. By Bezout's theorem \cite{IRS74}, the maximum number of solutions that system \eqref{equ3-1-9} can have is $2^{n-3}$. In the following, we claim that the bound can be reached.

Let $\mathcal{S}_0$ be the set of algebraic systems of the form \eqref{equ3-1-9}. We show that there exist many systems in $\mathcal{S}_0$ having exactly $2^{n-3}$ simple zeros. Consider a family of systems
\begin{align}\label{equ3-1-10-1}
&a_{1,1,0,\boldsymbol{0}_{n-2}}
+b_{1,0,1,\boldsymbol{0}_{n-2}}+(a_{1,0,\boldsymbol{e}_3}
+b_{0,1,\boldsymbol{e}_3})X_3=0, \\ \label{equ3-1-10-2}
&(c_{2,0,\boldsymbol{0}_{n-2},3}
+c_{0,2,\boldsymbol{0}_{n-2},3})R^2+2\sum_{3\leq j_1\leq j_2\leq n}c_{0,0,\boldsymbol{e}_{j_1j_2},3}X_{j_1}X_{j_2}+2\sum_{j=3}^nc_{1,0,0,\boldsymbol{e}_j,3}X_j=0,\\ \label{equ3-1-10-3}
&2\sum_{3\leq j_1\leq j_2\leq s}c_{0,0,\boldsymbol{e}_{j_1j_2},s}X_{j_1}X_{j_2}+2c_{1,0,0,\boldsymbol{e}_j,s}X_s=0,\quad s=4,\ldots,n,
\end{align}
with all the coefficients nonzero. Substituting the unique solution $X_{3,0}$ of $X_3$ in \eqref{equ3-1-10-1} into \eqref{equ3-1-10-3} with $s=4$, then the last equation has exactly two different real solutions $X_{4,1}$ and $X_{4,2}$ of $X_4$ by choosing appropriately the values of coefficients in \eqref{equ3-1-10-3}. Introducing the two solutions $(X_{3,0},X_{4,i})$, $i=1,2$, into \eqref{equ3-1-10-3} with $s=5$, and choosing appropriately the values of the resulting coefficients, we obtain two different real solutions $X_{5,i,1}$ and $X_{5,i,2}$ of $X_5$ for each $i=1,2$. Moreover, one can choose the coefficients so that the four solutions $(X_{3,0},X_{4,i},X_{5,i,j})$ for $i,j=1,2$ are distinct. By induction we can prove that for suitable choice of the coefficients equations \eqref{equ3-1-10-1} and \eqref{equ3-1-10-3} have $2^{n-3}$ different solutions $(X_3,X_4,\ldots,X_n)$. Note that $X_3=X_{3,0}$ is fixed, so for an appropriate choice of the coefficients
in \eqref{equ3-1-10-2}, the resulting equation can have a positive solution $R$ for each of the $2^{n-3}$ solutions $(X_3,X_4,\ldots,X_n)$ of \eqref{equ3-1-10-1} and \eqref{equ3-1-10-3}. Since the $2^{n-3}$ solutions are different, and $2^{n-3}$ is the maximum number that equations \eqref{equ3-1-10-1}-\eqref{equ3-1-10-3} can have by Bezout's theorem. we conclude that every solution is simple, and so the determinant of the Jacobian of the system evaluated at the solution is nonzero. This establishes the claim.

Using similar arguments, we can also choose the coefficients of the former system so that it has $\ell\in\{0,1,\ldots,2^{n-3}\}$ real simple solutions. Taking system \eqref{equ3-1-5} with $F_1$ and $F_s$ having the coefficients as in \eqref{equ3-1-10-1}-\eqref{equ3-1-10-3}, then the averaged functions in \eqref{equ3-1-8} have exactly $\ell\in\{0,1,\ldots,2^{n-3}\}$ real simple solutions with $R>0$. By the averaging theorem, we conclude that there are systems \eqref{equ-2} with a number $\ell\in\{0,1,\ldots,2^{n-3}\}$ limit cycles bifurcating from the origin. This completes the proof of Theorem \ref{main-th-1}.

\section{Demonstration of the Maple Program}\label{sect-Algo}

This section reports the performance of the function \textsc{OrderKFormula}$(k,n)$ (see Section \ref{sect3.1}) on several examples.

\begin{table}[h]
\begin{center}
\caption{Computational times (in seconds) of the function \textsc{OrderKFormula}$(k,n)$.}\label{Tab-A}
\begin{tabular}{|c|c|c||c||c||c||c|}
  \hline
  \multicolumn{7}{|c|}{$k$}  \\
  \hline \hline
  & & 1 & 2 & 3 & 4 & 5 \\
  \hline \hline
  \multirow{6}{*}{$n$}
  & 2 & 0. & 0.024 & 0.034 & 0.055 & 0.107 \\
  & 3 & 0. & 0.027 & 0.052 & 0.137 & 0.685 \\
  & 4 & 0. & 0.031 & 0.076 & 0.399 & 4.071 \\
  & 5 & 0. & 0.033 & 0.125 & 1.033 & 21.699 \\
  & 6 & 0. & 0.039 & 0.203 &  2.563 & 103.486 \\
  & 7 & 0.002 & 0.046 & 0.322 & 5.527 & 440.564 \\
 % & 8 & 3 & 9 & 27 & 81 & - & - &  \\
 % & $\vdots$ &  &  &  &  &  & $\ddots$\\
  \hline
\end{tabular}
\end{center}
\end{table}

The output \textsc{OrderKFormula}$(k,2)$ for $k=1,2,3$ is:
$[y_{k,1}(t,z_{{1}},z_{{2}}),y_{k,2}(t,z_{{1}},z_{{2}})]$, where
\begin{align*}
y&_{1,1}(t,z_{{1}},z_{{2}})=\int_{0}^{t}F_{{1,1}}\left(\theta,z_{{1}},z_{{2
}}\right)d\theta,\\
y&_{1,2}(t,z_{{1}},z_{{2}})=\int_{0}^{t}F_{{1,2}}\left(
\theta,z_{{1}},z_{{2}} \right)d\theta,\\
y&_{2,1}(t,z_{{1}},z_{{2}})=\int_{0}^{t}2\,F_{{2,1}}\left(\theta,z_{{1}},z_
{{2}}\right)+2\,{\frac{\partial}{\partial z_{{1}}}}
F_{{1,1}}\left(\theta,z_{{1}},z_{{2}}\right)
\times y_{{1,1}}\left(\theta,z_{{1}},z_{{2}}\right)\\
&+2\,{\frac{\partial}{\partial z_{{2}}}}F_{{1,1}}\left(\theta,z_{{1}},z_{{2}}\right)
y_{{1,2}}\left(\theta,z_{{1}},z_{{2}}\right)d\theta,\\
y&_{2,2}(t,z_{{1}},z_{{2}})=\int_{0}^{t}2\,F_{{2,2}}\left(\theta,z_{{1}},z_{
{2}}\right)+2\,{\frac{\partial}{\partial z_{{1}}}}
F_{{1,2}}\left(\theta,z_{{1}},z_{{2}}\right)
\times y_{{1,1}}\left(\theta,z_{{1}},z_{{2}}\right)\\
&+2\,{\frac{\partial}{\partial z_{{2}}}}F_{{1,2}}
\left(\theta,z_{{1}},z_{{2}}\right)y_{{1
,2}}\left(\theta,z_{{1}},z_{{2}}\right)d\theta,\\
y&_{3,1}(t,z_{{1}},z_{{2}})=\int_{0}^{t}6\,F_{{3,1}}
\left(\theta,z_{{1}},z_{{2}}\right)+6\,{\frac {\partial }{\partial z_{{1}}}} F_{{2,1}}\left(\theta,z_{{1}},z_{{2}}\right)
\times y_{{1,1}}\left(\theta,z_{{1}},z_{{2}}\right)\\
&+6\,{\frac{\partial}{\partial z_{{2}}}}F_{{2,1}}\left(\theta,z_{{1}},z_{{2}}\right)y_{{1
,2}}\left(\theta,z_{{1}},z_{{2}}\right)
+3\,{\frac {\partial }{\partial z_{{1}}}}F_{{1,1}}\left(
\theta,z_{{1}},z_{{2}}\right)y_{{2,1}}
\left(\theta,z_{{1}},z_{{2}}\right)\\
&+3\,{\frac {\partial }{
\partial z_{{2}}}}F_{{1,1}}\left(\theta,z_{{1}},z_{
{2}}\right)
\times y_{{2,2}}\left(\theta,z_{{1}},z_{{2}}\right)+3\,
{\frac {\partial ^{2}}{\partial {z_{{1}}}^{2}}}F_{{1,1}}
\left(\theta,z_{{1}},z_{{2}}\right)y_{{1,1}}\left(\theta,z_{
{1}},z_{{2}}\right)^{2}\\
&+6\,{\frac {\partial^{2}}{
\partial z_{{1}}\partial z_{{2}}}}F_{{1,1}}\left(
\theta,z_{{1}},z_{{2}}\right)y_{{1,1}}
\left(\theta,z_{{1}},z_{{2}}\right)y_{{1,2}}
\left(\theta,z_{{1}},z_{{2}}\right)\\
&+3\,{\frac {\partial^{2}}{\partial z_{{2}}^{2}}}F_{{1,1}}
\left(\theta,z_{{1}},z_{{2}}\right)y_{{1,2}}
\left(\theta,z_{{1}},z_{{2}}\right)^{2}d\theta,\\
y&_{3,2}(t,z_{{1}},z_{{2}})=\int_{0}^{t}6\,F_{{3,2}}\left(\theta,z_{{1}},z_{
{2}}\right)+6\,{\frac {\partial }{\partial z_{{1}}}}
F_{{2,2}}\left(\theta,z_{{1}},z_{{2}}\right)
\times y_{{1,1}}\left(\theta,z_{{1}},z_{{2}}\right)\\
&+6\,{\frac {\partial}{\partial z_{{2}}}}F_{{2,2}}
\left(\theta,z_{{1}},z_{{2}}\right)y_{{1
,2}}\left(\theta,z_{{1}},z_{{2}}\right)
+3{\frac {\partial }{\partial z_{{1}}}}F_{{1,2}}\left(
\theta,z_{{1}},z_{{2}}\right)y_{{2,1}}\left(\theta,z_{{1}},z_{{2}}\right)\\
&+3\,{\frac {\partial}{\partial z_{{2}}}}F_{{1,2}}\left(\theta,z_{{1}},z_{
{2}}\right)
\times y_{{2,2}}\left(\theta,z_{{1}},z_{{2}}\right)
+3\,{\frac {\partial^{2}}{\partial z_{{1}}^{2}}}F_{{1,2}}\left( \theta,z_{{1}},z_{{2}}\right)y_{{1,1}}
\left(\theta,z_{{1}},z_{{2}}\right)^{2}\\
&+6\,{\frac {\partial^{2}}{
\partial z_{{1}}\partial z_{{2}}}}F_{{1,2}}\left(
\theta,z_{{1}},z_{{2}}\right)y_{{1,1}}
\left(\theta,z_{{1}},z_{{2}}\right)y_{{1,2}}
\left(\theta,z_{{1}},z_{{2}}\right)\\
&+3\,{\frac {\partial^{
2}}{\partial z_{{2}}^{2}}}F_{{1,2}}\left(\theta,z
_{{1}},z_{{2}}\right)y_{{1,2}}
\left( \theta,z_{{1}},z_{{2}}\right)^{2}d\theta.
\end{align*}

\section{Proof of Theorem \ref{main-2}}\label{sect-chun}
We need to write the linear part of system \eqref{equ3-7-2} at the origin in its real Jordan normal form
\begin{equation}\label{equros-1}
\begin{split}
\left(
\begin{array}{ccc}
0&-\beta&0\\
\beta&0&0\\
0&0&0\\
\end{array}
\right),
\end{split}
\end{equation}
when $\varepsilon=0$. For doing that we do the linear change of variables $(x,y,z)\rightarrow(u,v,w)$ given by
\begin{equation}\label{equros-1-1}
\begin{split}
x=\frac{v}{\beta}+w,\quad
y=u,\quad z=-\beta v.\nonumber
\end{split}
\end{equation}
In these new variables $(u,v,w)$, system \eqref{equ3-7-2} becomes a new system which can be written as $(\dot{u},\dot{v},\dot{w})$. Computing the third-order Taylor expansion of expressions in this new system, with respective to $\varepsilon$, about the point $\varepsilon=0$, we obtain
\begin{equation}\label{equros-2}
\begin{split}
\dot{u}&=-\beta v,\\
\dot{v}&=\beta u+3\,\frac {{v}^{2}w}{{\beta}^{3}}+\frac {{v}^{3}}{{\beta}^{4}}-
\frac{{u}^{2}v}{{\beta}^{2}}+3\,\frac {v{w}^{2}}{{\beta}^{2}}-
\frac{{u}^{2}w}{\beta}\\
&\quad+\frac{{w}^{3}}{\beta}+\varepsilon\Big(\Big(-a_{{1}}+\frac{b_{{1}}}{{\beta}^{2}}\Big)v-\frac {uc_{{1}}}{\beta}+\frac{wb_{{1}}}{\beta}\Big)\\
&\quad+\varepsilon^2\Big(\Big(\frac {b_{{2}}}{{\beta}^{2}}-a_{{2}}\Big)v-\frac {uc_{{2}}}{\beta}+\frac{wb_{{2}}}{\beta}\Big),\\
\dot{w}&=-\frac{{v}^{3}}{{\beta}^{5}}+\frac{{u}^{2}v}{{\beta}^{3}}-
\frac{{w}^{3}}{{\beta}^{2}}-3\frac {v{w}^{2}}{{\beta}^{3}}+
\frac{{u}^{2}w}{{\beta}^{2}}-3\frac {{v}^{2}w}{{\beta}^{4}}\\
&\quad+\varepsilon\Big(\Big(\frac {a_{{1}}}{\beta}-\frac {b_{{1}}}{{\beta}^{3}}
\Big)v+\frac {uc_{{1}}}{{\beta}^{2}}-\frac {wb_{{1}}}{{\beta}^{2}}\Big)\\
&\quad+\varepsilon^2\Big(\Big(-\frac{b_{{2}}}{{\beta}^{3}}+\frac
{a_{{2}}}{\beta}\Big)v+\frac {uc_{{2}}}{{\beta}^{2}}-\frac{wb_{{2}}}{{\beta}^{
2}}\Big).
\end{split}
\end{equation}
Note that system \eqref{equros-2} is in the form of \eqref{equ-2}. Now using our Maple program to system \eqref{equros-2}, we obtain the first-order averaged functions:
\begin{equation}\label{equros-3}
\begin{split}
f_{1,1}(R,X_3)=-\frac{\pi R\big(\beta^2a_1-b_1\big)}{\beta^3},\quad
f_{1,3}(R,X_3)=-\frac{2b_1\pi X_3}{\beta^3}.\nonumber
\end{split}
\end{equation}
It obvious that this algebraic system has no real isolated solution with $R>0$. Hence, the first-order averaging does not provide any information about the limit cycles that bifurcate from the origin of system \eqref{equ3-7-2}.

We pass then to the second-order averaging, assuming $(f_{1,1}(R,X_3),f_{1,3}(R,X_3))=(0,0)$. This makes $a_1=b_1=0$. Updating the obtained normal form and computing the second-order averaged functions, we have
\begin{equation}\label{equros-6}
\begin{split}
f_{2,1}(R,X_3)&=-\frac{\pi R}{4\beta^5}\bar{f}_{2,1}(R,X_3),\\
f_{2,3}(R,X_3)&=\frac{\pi}{\beta^5}\bar{f}_{2,3}(R,X_3),
\end{split}
\end{equation}
where
\begin{equation}\label{equros-6-1}
\begin{split}
\bar{f}_{2,1}&(R,X_3)=\left(\beta^{2}-3\right)\rho+4\beta^{4}a_{{2}}
-12\beta^{2}X_{{3}}^{2}-4\beta^{2}b_{{2}},\\
\bar{f}_{2,3}&(R,X_3)=X_3\left(\left(\beta^{2}-3\right)\rho-2\beta^{2}X_{{3}}^{2}
-2\beta^{2}b_{{2}}\right),\nonumber
\end{split}
\end{equation}
with $\rho=R^2$. A direct computation shows that the BKK bound of the polynomial system $\{\bar{f}_{2,1}(R,X_3),\bar{f}_{2,2}(R,X_3)\}$ is 3. So system \eqref{equros-6} can have at most 3 real solutions with $\rho>0$. As a result, system \eqref{equ3-7-2} can have at most 3 limit cycles bifurcating from the origin. In the following we compute a partition of the parametric space such that, inside every open cell of the partition, the system can have 3 limit cycles.

The determinate of the Jacobian of $(f_{2,1}(R,X_3),f_{2,3}(R,X_3))$ is:
\begin{equation}\label{equros-7}
\begin{split}
D_2(R,X_3)=\mbox{det}\left(
 \begin{matrix}
   \frac{\partial f_{2,1}}{\partial R} & \frac{\partial f_{2,1}}{\partial X_3}\\
   \frac{\partial f_{2,3}}{\partial R} & \frac{\partial f_{2,3}}{\partial X_3}
  \end{matrix}
  \right)
  =-\frac{\pi^2}{4\beta^{10}}\cdot\bar{D}_2(R,X_3),\nonumber
\end{split}
\end{equation}
where
\begin{equation}\label{equros-7-1}
\begin{split}
\bar{D}_2&(R,X_3)=\left(3\beta^{4}-18\beta^{2}+27\right)\rho^{2}
+\left(18\beta^{4}-54\beta^{2}\right)\rho X_{{3}}^{2}\\
&\quad+\left(4\beta^{6}a_{{2}}-12\beta^{4}a_{{2}}-10\beta^{4}b_{{2}}+30
\beta^{2}b_{{2}}\right)\rho+72\beta^{4}X_{{3}}^{4}\\
&\quad+\left(-24\beta^{6}a_{{2}}+48\beta^{4}b_{{2}}\right)X_{{3}}^{2}
-8\beta^{6}a_{{2}}b_{{2}}+8\beta^{4}b_{{2}}^{2},\nonumber
\end{split}
\end{equation}
with $\rho=R^2$.

By Theorem \ref{semi-averaging}, we know that system \eqref{equ3-7-2} can have 3 limit cycles bifurcating from the origin if the following semi-algebraic system
\begin{equation}\label{equros-8}
\begin{split}
\left\{
\begin{array}{ll}
&\bar{f}_{2,1}(R,X_3)=\bar{f}_{2,3}(R,X_3)=0, \\
&\rho>0,\quad \bar{D}_2(R,X_3)\neq0,\quad \beta\neq0
\end{array}
\right.
\end{split}
\end{equation}
has exactly 3 real solutions with respective to
the variables $R$, $X_3$. Using {\sf{DISCOVERER}} (or the package RegularChains[SemiAlgebraicSetTools] in Maple), we obtain system \eqref{equros-8} has exactly 3 real solutions if and only if the condition $\mathcal{C}_0$ or the condition $\mathcal{C}_1$ holds (see \eqref{equ3-7-3}).

This completes the proof of Theorem \ref{main-2}.

\section{Other Models and Remarks}\label{sect-ex}

\subsection{A Class of Generalized Lorenz Systems}\label{sect4.2-1}

Consider the following integrable deformation of Lorenz system:
\begin{equation}\label{equ4-2-1}
\begin{split}
\dot{x}&=a(y-x)+dy(z-c),\\
\dot{y}&=cx-xz-y,\\
\dot{z}&=-bz+xy+sx,
\end{split}
\end{equation}
where $a$, $b$, $c$, $d$, $s$ are real parameters. System \eqref{equ4-2-1} is obtained by choosing the deformation functions $\alpha=\frac{1}{2}dy^2$ and $\beta=sy$ (see \cite[Section 3.1]{crla2018}). We notice that if $d=s=0$, then this system is the so-called Lorenz System. It is easy to check that the origin is a zero-Hopf equilibrium when $a=-1$, $b=0$ and $c^2d+c-1>0$. Now consider the vector $(a,b,c,d,s)$ given by
\begin{equation}\label{equ4-2-2}
\begin{split}
a&=-1+\varepsilon a_1+\varepsilon^2 a_2+\varepsilon^3 a_3,\quad
b=\varepsilon b_1+\varepsilon^2 b_2+\varepsilon^3 b_3,\\
c&=c_0+\varepsilon c_1+\varepsilon^2 c_2+\varepsilon^3 c_3,\quad
s=s_0+\varepsilon s_1+\varepsilon^2 s_2+\varepsilon^3 s_3,\\
d&=\frac{\beta^2-c_0+1}{c_0^2}+\varepsilon d_1+\varepsilon^2 d_2+\varepsilon^3 d_3,\quad \beta>0,\nonumber
\end{split}
\end{equation}
where the constants $a_i$, $b_i$, $c_i$, $d_i$ and $s_i$ are all real parameters with $c_0\neq0$. Applying the third-order averaging method to system \eqref{equ4-2-1}, we have the following result.
\begin{theorem}\label{main-3}
The following statements hold for $|\varepsilon|>0$ sufficiently small.
\begin{itemize}
\item [(i)] The first-order averaging does not provide any
information about limit cycles that bifurcate from the origin.
\item [(ii)] System \eqref{equ4-2-1} has, up to the second-order averaging, at most 1 limit cycle bifurcating from the origin, and this number can be reached if one of the following 2 conditions holds:
\begin{equation}\label{equ4-2-3}
\begin{split}
\mathcal{C}_2&=[a_2<0,\,2\beta^2+2-c_0<0]\wedge[\beta>0],\\ \mathcal{C}_3&=[0<a_2,\,0<2\beta^2+2-c_0]\wedge[\beta>0,c_0\neq0].
\end{split}
\end{equation}
\item [(iii)] System \eqref{equ4-2-1} has, up to the third-order averaging, at most 3 limit cycles bifurcating from the origin, and this number can be reached if we take the condition $\mathcal{C}^*=[c_1=d_1=1,d_2=s_2=2]$ and the sample points of $(b_3,a_3,\delta)$ listed in Table \ref{Tab-2}, where $\delta=\sqrt{\beta^2+1}$.
\end{itemize}
\end{theorem}

\begin{remark}\label{remmark-Lorenz}
Remark that, when applying the third-order averaging to system \eqref{equ4-2-1}, we found that the resulting semi-algebraic system \eqref{equ3-0-4} contains 7 parameters that our algorithmic approach cannot work effectively (the computation in Maple was consuming too much of the CPU). On the other hand, the information provided by sample points of the parameter space may often be sufficient in practice. In general, the selection of sample points might be extremely complicated but could be automated using, e.g., the Maple command \textit{RealRootClassification} (with the option \textit{'output'='samples'}) to the semi-algebraic system. In order to obtain sample points for system \eqref{equ4-2-1} to have 3 real solutions, we restrict parameter condition $\mathcal{C}^*=[c_1=d_1=1,d_2=s_2=2]$. The command \textit{RealRootClassification} permits us to select 1750 sample points from the resulting semi-algebraic system, we show only 60 sample points in Table \ref{Tab-2} due to the space limitation.
\end{remark}

\subsection{A 4D Hyperchaotic System}\label{sect4.3-1}

Recently, the following new hyperchaotic system is
proposed in \cite{zyxc2019}:
\begin{equation}\label{equ4-3-1}
\begin{split}
\dot{x}&=a_1(y-x)-w,\\
\dot{y}&=a_2x-xz-y,\\
\dot{z}&=xy-a_3z,\\
\dot{w}&=a_4xz-a_5w,
\end{split}
\end{equation}
where $a_i$'s are real parameters. One can verify that the origin is a (complete) zero-Hopf equilibrium when $a_1=-1$, $a_3=a_5=0$ and $a_2-1>0$. Consider the vector $(a_1,a_2,a_3,a_4,a_5)$ given by
\begin{equation}\label{equ4-3-2}
\begin{split}
a_1&=-1+\varepsilon a_{1,1},\quad
a_2=1+\beta^2+\varepsilon a_{2,1},\\
a_3&=\varepsilon a_{3,1},\quad
a_4=a_{4,0}+\varepsilon a_{4,1},\\
a_5&=\varepsilon a_{5,1},\quad \beta\neq0,\nonumber
\end{split}
\end{equation}
where the constants $a_{i,j}$ are all real parameters. The following result provides sufficient conditions for the bifurcation of a limit cycle from the origin.

\begin{theorem}\label{main-4}
For $|\varepsilon|>0$ sufficiently small, using the first-order averaging method, we obtain that, system \eqref{equ4-3-1} has exactly 1 limit cycle bifurcating from the origin if one of the following 4 conditions holds:
\begin{equation}\label{equ4-3-3}
\begin{split}
\mathcal{C}_4&=[a_{4,0}<0,\,a_{3,1}<0,\,0<a_{1,1}]\wedge\tilde{\mathcal{C}},\\ \mathcal{C}_5&=[a_{4,0}<0,\,0<a_{3,1},\,a_{1,1}<0]\wedge\tilde{\mathcal{C}},\\
\mathcal{C}_6&=[0<a_{4,0},\,a_{3,1}<0,\,a_{1,1}<0]\wedge\tilde{\mathcal{C}},\\ \mathcal{C}_7&=[0<a_{4,0},\,0<a_{3,1},\,0<a_{1,1}]\wedge\tilde{\mathcal{C}},\\
\end{split}
\end{equation}
where $\tilde{\mathcal{C}}=[\beta\neq0,a_{4,0}\neq0,a_{3,1}\neq0,a_{1,1}\neq0,a_{1,1}+a_{5,1}\neq0]$ is of type border polynomial.
\end{theorem}

Theorem \ref{main-3} and Theorem \ref{main-4} can be proved by using similar calculations and arguments to the proof of Theorem \ref{main-2}. The details of their proof are omitted here.

\renewcommand\arraystretch{1.5}
\begin{sidewaystable}
%\begin{table}[h]
\begin{center}
\caption{Selected sample points of $(b_3,a_3,\delta)$ for system \eqref{equ4-2-1} to have 3 limit cycles.}\label{Tab-2}
\begin{tabular}{l|l|l}
\toprule[1pt]
  \multicolumn{3}{c}{60 sample points of $(b_3,a_3,\delta)$ with $\delta=\sqrt{\beta^2+1}$} \\ \midrule[1pt]
$b_{{3}}={\frac {3}{512}},a_{{3}}=-{\frac {461}{16384}},\delta={\frac
{17477}{16384}}$ & $b_{{3}}={\frac {3}{512}},a_{{3}}=-{\frac {117}{8192}},\delta={\frac {
34903}{32768}}
$ & $b_{{3}}={\frac {3}{512}},a_{{3}}=-{\frac {117}{8192}},\delta={\frac {
17465}{16384}}$ \\
  \hline
 $b_{{3}}={\frac {3}{512}},a_{{3}}=-{\frac {159}{16384}},\delta={\frac
{8709}{8192}}
$ & $b_{{3}}={\frac {3}{512}},a_{{3}}=-{\frac {159}{16384}},\delta={\frac
{8725}{8192}}
$  & $b_{{3}}={\frac {3}{512}},a_{{3}}=-{\frac {159}{16384}},\delta={\frac
{8733}{8192}}$ \\
  \hline
 $b_{{3}}={\frac {3}{512}},a_{{3}}={\frac {127}{32768}},\delta={\frac {
8709}{8192}}
$ &  $b_{{3}}={\frac {3}{512}},a_{{3}}={\frac {127}{32768}},\delta={\frac {
8725}{8192}}
$ &  $b_{{3}}={\frac {3}{512}},a_{{3}}={\frac {127}{32768}},\delta={\frac {
8733}{8192}}
$ \\
  \hline
 $b_{{3}}={\frac {3}{512}},a_{{3}}={\frac {69}{8192}},\delta={\frac {
34903}{32768}}
$ & $b_{{3}}={\frac {3}{512}},a_{{3}}={\frac {69}{8192}},\delta={\frac {
17465}{16384}}
$  & $b_{{3}}={\frac {3}{512}},a_{{3}}={\frac {365}{16384}},\delta={\frac {
17477}{16384}}
$ \\
 \hline
 $b_{{3}}={\frac {1739}{32768}},a_{{3}}=-{\frac {795}{4096}},\delta={
\frac {34939}{32768}}
$ &  $b_{{3}}={\frac {1739}{32768}},a_{{3}}=-{\frac {281}{2048}},\delta={
\frac {8729}{8192}}
$ & $b_{{3}}={\frac {1739}{32768}},a_{{3}}=-{\frac {281}{2048}},\delta={
\frac {8733}{8192}}
$ \\

 \hline
 $b_{{3}}={\frac {1739}{32768}},a_{{3}}=-{\frac {345}{4096}},\delta={
\frac {1087}{1024}}
$ &  $b_{{3}}={\frac {1739}{32768}},a_{{3}}=-{\frac {345}{4096}},\delta={
\frac {8725}{8192}}
$ & $b_{{3}}={\frac {1739}{32768}},a_{{3}}=-{\frac {345}{4096}},\delta={
\frac {8733}{8192}}
$ \\

 \hline
 $b_{{3}}={\frac {1739}{32768}},a_{{3}}={\frac {127}{4096}},\delta={
\frac {1087}{1024}}
$ &  $b_{{3}}={\frac {1739}{32768}},a_{{3}}={\frac {127}{4096}},\delta={
\frac {8725}{8192}}
$ & $b_{{3}}={\frac {1739}{32768}},a_{{3}}={\frac {127}{4096}},\delta={
\frac {8733}{8192}}
$ \\

 \hline
 $b_{{3}}={\frac {1739}{32768}},a_{{3}}={\frac {43}{512}},\delta={
\frac {8729}{8192}}
$ &  $b_{{3}}={\frac {1739}{32768}},a_{{3}}={\frac {43}{512}},\delta={
\frac {8733}{8192}}
$ & $b_{{3}}={\frac {1739}{32768}},a_{{3}}={\frac {9}{64}},\delta={\frac {
34939}{32768}}
$ \\

 \hline
 $b_{{3}}={\frac {51843751309}{549755813888}},a_{{3}}=-{\frac {1429}{
4096}},\delta={\frac {279481}{262144}}
$ &  $b_{{3}}={\frac {51843751309}{549755813888}},a_{{3}}=-{\frac {275}{
1024}},\delta={\frac {279443}{262144}}
$ & $b_{{3}}={\frac {51843751309}{549755813888}},a_{{3}}=-{\frac {275}{
1024}},\delta={\frac {34933}{32768}}
$ \\

  \hline
 $b_{{3}}={\frac {51843751309}{549755813888}},a_{{3}}=-{\frac {313}{
2048}},\delta={\frac {2175}{2048}}
$ &  $b_{{3}}={\frac {51843751309}{549755813888}},a_{{3}}=-{\frac {313}{
2048}},\delta={\frac {17871791}{16777216}}
$ & $b_{{3}}={\frac {51843751309}{549755813888}},a_{{3}}=-{\frac {313}{
2048}},\delta={\frac {17885025}{16777216}}
$ \\

 \hline
 $b_{{3}}={\frac {51843751309}{549755813888}},a_{{3}}=-{\frac {207}{
2048}},\delta={\frac {8565}{8192}}
$ &  $b_{{3}}={\frac {51843751309}{549755813888}},a_{{3}}=-{\frac {207}{
2048}},\delta={\frac {4575160993}{4294967296}}
$ & $b_{{3}}={\frac {51843751309}{549755813888}},a_{{3}}=-{\frac {195}{
2048}},\delta={\frac {542193}{524288}}
$ \\

  \hline
 $b_{{3}}={\frac {51843751309}{549755813888}},a_{{3}}=-{\frac {195}{
2048}},\delta={\frac {285947481}{268435456}}
$ &  $b_{{3}}={\frac {51843751309}{549755813888}},a_{{3}}=-{\frac {195}{
2048}},\delta={\frac {4578528035}{4294967296}}
$ & $b_{{3}}={\frac {51843751309}{549755813888}},a_{{3}}={\frac {71}{65536
}},\delta={\frac {271227}{262144}}
$ \\

  \hline
 $b_{{3}}={\frac {51843751309}{549755813888}},a_{{3}}={\frac {71}{65536
}},\delta={\frac {9150319465}{8589934592}}
$ &  $b_{{3}}={\frac {51843751309}{549755813888}},a_{{3}}={\frac {71}{65536
}},\delta={\frac {9157056213}{8589934592}}
$ & $b_{{3}}={\frac {51843751309}{549755813888}},a_{{3}}={\frac {115}{
16384}},\delta={\frac {1071}{1024}}
$ \\

  \hline
 $b_{{3}}={\frac {51843751309}{549755813888}},a_{{3}}={\frac {115}{
16384}},\delta={\frac {4575161051}{4294967296}}
$ &  $b_{{3}}={\frac {51843751309}{549755813888}},a_{{3}}={\frac {477}{8192
}},\delta={\frac {2175}{2048}}
$ & $b_{{3}}={\frac {51843751309}{549755813888}},a_{{3}}={\frac {477}{8192
}},\delta={\frac {17871791}{16777216}}
$ \\

  \hline
 $b_{{3}}={\frac {51843751309}{549755813888}},a_{{3}}={\frac {477}{8192
}},\delta={\frac {17885025}{16777216}}
$ &  $b_{{3}}={\frac {51843751309}{549755813888}},a_{{3}}={\frac {89}{512}}
,\delta={\frac {279443}{262144}}
$ & $b_{{3}}={\frac {51843751309}{549755813888}},a_{{3}}={\frac {89}{512}}
,\delta={\frac {558927}{524288}}
$ \\

\hline
 $b_{{3}}={\frac {51843751309}{549755813888}},a_{{3}}={\frac {65}{256}}
,\delta={\frac {279481}{262144}}$ &  $b_{{3}}={\frac {25922055879}{274877906944}},a_{{3}}=-{\frac {1429}{
4096}},\delta={\frac {279481}{262144}}$ & $b_{{3}}={\frac {25922055879}{274877906944}},a_{{3}}=-{\frac {275}{
1024}},\delta={\frac {279443}{262144}}$ \\

\hline
 $b_{{3}}={\frac {25922055879}{274877906944}},a_{{3}}=-{\frac {275}{
1024}},\delta={\frac {34933}{32768}}
$ &  $b_{{3}}={\frac {25922055879}{274877906944}},a_{{3}}=-{\frac {5}{32}},
\delta={\frac {2175}{2048}}
$ & $b_{{3}}={\frac {25922055879}{274877906944}},a_{{3}}=-{\frac {5}{32}},
\delta={\frac {4467949}{4194304}}$ \\

\hline
 $b_{{3}}={\frac {25922055879}{274877906944}},a_{{3}}=-{\frac {5}{32}},
\delta={\frac {8942519}{8388608}}$ &  $b_{{3}}={\frac {25922055879}{274877906944}},a_{{3}}=-{\frac {53}{512}
},\delta={\frac {2147}{2048}}$ & $b_{{3}}={\frac {25922055879}{274877906944}},a_{{3}}=-{\frac {53}{512}
},\delta={\frac {142973799}{134217728}}$ \\

\hline
 $b_{{3}}={\frac {25922055879}{274877906944}},a_{{3}}={\frac {77}{8192}
},\delta={\frac {8589}{8192}}$ &  $b_{{3}}={\frac {25922055879}{274877906944}},a_{{3}}={\frac {77}{8192}
},\delta={\frac {285947601}{268435456}}$ & $b_{{3}}={\frac {25922055879}{274877906944}},a_{{3}}={\frac {253}{4096
}},\delta={\frac {2175}{2048}}$ \\

\hline
 $b_{{3}}={\frac {25922055879}{274877906944}},a_{{3}}={\frac {253}{4096
}},\delta={\frac {17871797}{16777216}}$ &  $b_{{3}}={\frac {25922055879}{274877906944}},a_{{3}}={\frac {253}{4096
}},\delta={\frac {17885037}{16777216}}$ & $b_{{3}}={\frac {25922055879}{274877906944}},a_{{3}}={\frac {89}{512}}
,\delta={\frac {279443}{262144}}$ \\
\bottomrule[1pt]
\end{tabular}
\end{center}
%\end{table}
\end{sidewaystable}

\end{document}